\documentclass{article}

\usepackage{caption}
\usepackage{subcaption}
\usepackage[utf8]{inputenc}
\usepackage[T1]{fontenc}
\usepackage[english]{babel}
\usepackage{hyperref} 
\usepackage{setspace} 
\usepackage{enumitem}
\usepackage{amsthm}
\usepackage{amsmath}
\usepackage{stmaryrd}
\usepackage{amssymb}
\usepackage{mathrsfs}
\usepackage{tikz-cd}
\usepackage{tikz}
\usepackage{ebproof}
\usepackage[nameinlink, capitalize]{cleveref}
\usetikzlibrary{shapes}

\usepackage{biblatex}
\definecolor{blue(pigment)}{rgb}{0.2, 0.2, 0.6}
\definecolor{darkgreen}{rgb}{0.0, 0.5, 0.0}
\definecolor{darkred}{HTML}{9F000F}
\hypersetup
{
    colorlinks=true,
    linkcolor=darkred, 
    filecolor=magenta,            
    urlcolor=blue(pigment),
    citecolor=darkgreen
}

\usepackage{geometry}
\setitemize[0]{itemsep = 0ex, topsep = 0.25em}

\usepackage{InternshipNotations}

\usepackage{subfiles}

\title{Coherent differentiation in models of  \newline \newline Linear Logic}
\author{Aymeric Walch,\ \url{aymeric.walch@ens-lyon.fr}}
\date{June, 2022}

\begin{document}

\makeatletter
 \begin{titlepage}
  \centering
      \includegraphics[width=0.4\textwidth]{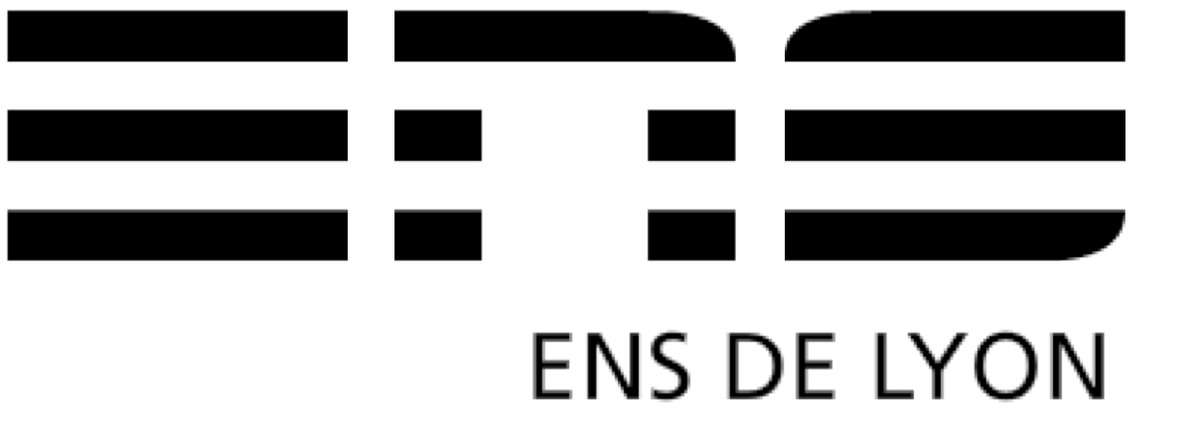}
      \hfill
    \includegraphics[width=0.2\textwidth]{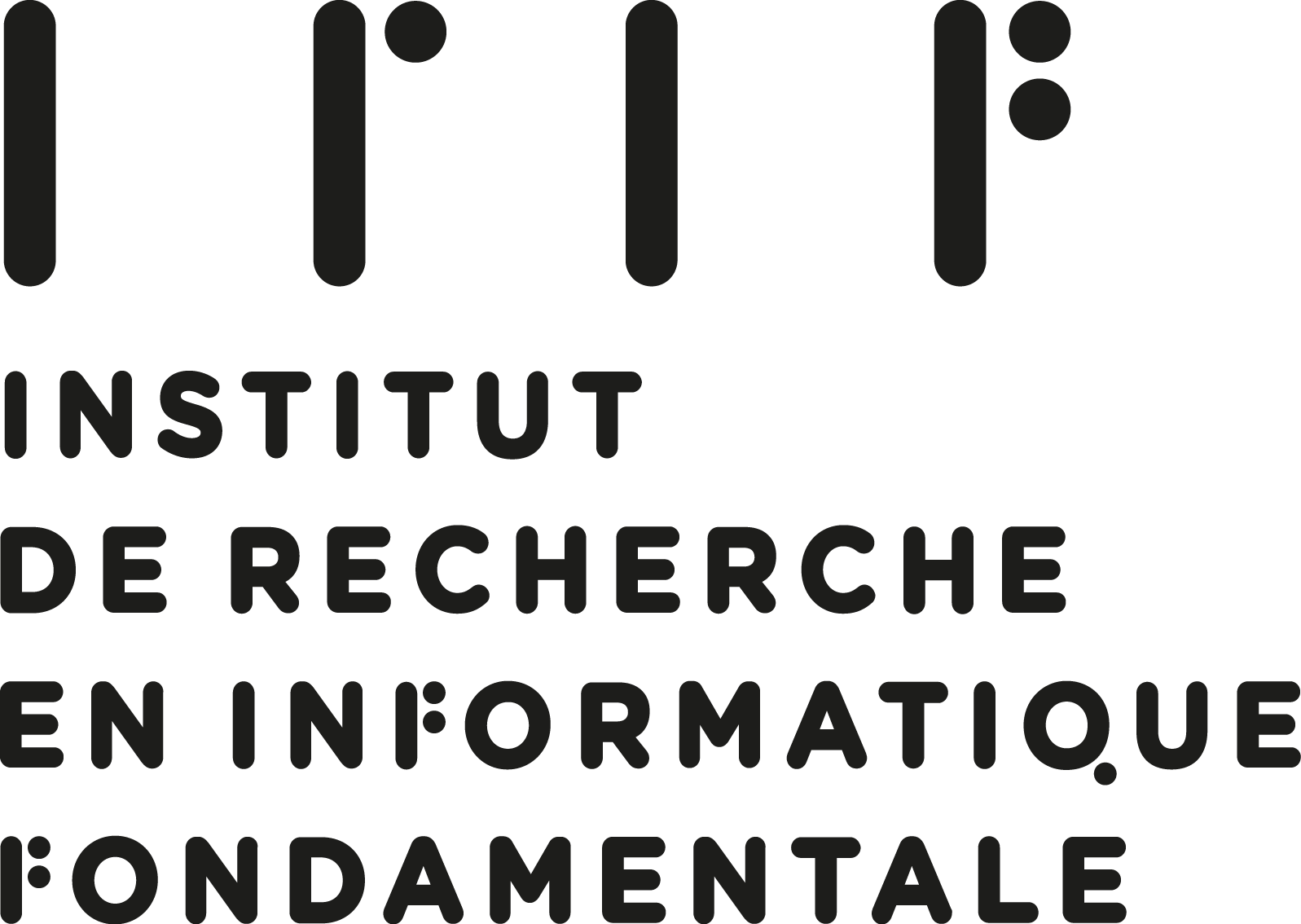} \\
    \vfill
    \vspace*{-2em}
         { \Huge \@title} \\
    \vspace{2em}
      { \large \textbf {\@author}} \\
      { \large \textsc{ENS de Lyon}} \\
      { \large \textsc{M2IF}} \\
    \vspace{1em}
    { \large { \@date }}
    \vfill
      { \large{ Under the supervision of Thomas Ehrhard} }

  \end{titlepage}

Proof theory is a field whose objects of study are the proof themselves. A fundamental concept of proof theory is the Curry Howard isomorphism that relates programs with proof in three levels. The types of a programming language relates to formulas. A program of type $A \arrow B$ relates to a proof of the formula $A \imply B$. Finally, program execution (namely the $\beta$-reduction) relates to a notion of computation on proof called the \emph{cut elimination}. This observation introduced a deeply influential idea: proof theory is not only about the \emph{provability} of formulas, but also about the \emph{computational content} of their proofs. Two proofs of the same formula can relate deeply between each other, or conversely have nothing to do in common.

The issue is that computation is usually very hard to describe and study, as many operational properties can interact between each other in a very subtle way. Thus designing a proof system (or a calculus) usually comes with a lot of requirements, such as proving subject reduction (computation conserves typing), normalization (computation terminates), confluence (different computational strategy converges to the same result), etc. Denotational semantic is a field that arose in order to simplify those issues. The goal of this field is to to describe \emph{what} a program computes while being agnostic on \emph{how} it computes, by giving some object $\mathcal{M}$ called a \emph{model} and an interpretation $\semantic{.}: \{\text{proofs/programs}\} \arrow \mathcal{M} $  invariant under computation. People then noticed that those models should actually have the structure of categories, hence turning the field of denotational semantic into the field of \emph{categorical semantic}. This field found quite a lot of success, for the reason mentioned above but also in more concrete setting: categorical semantic gives a way to describe the ``meaning'' of a program, something that is quite crucial when one wants to prove that the program behaves correctly.

Categorical semantic has another crucial application. Sometimes, some properties that are not captured by the syntax might be shared between many different models. It might suggest  some new syntactical ideas that should be introduced to account for those. This is what happened when Girard \cite{Girard_1987} discovered \emph{linear logic}, a refinement of intuitionistic and classical logic in which proofs and morphisms are interpreted through the lenses of linear algebra. This is also what happened when my supervisor Thomas Ehrhard \cite{Ehrhard_2003} built upon the ideas of linear logic and discovered that many models had a notion of ``linear approximation'' akin to differentiation, hence leading to the development of \emph{differential linear logic}. This is finally what might be happening again with the very recent introduction of \emph{coherent differentiation} also by Ehrhard \cite{Ehrhard21}. He noticed that many models of linear logic were not models of differential linear logic, yet still admitted a notion of differentiation, suggesting that differential linear logic could be improved.

The goal of this internship was to work in the development of this exciting new concept. Concretely, the internship was split in three part of one month and a half each. The first part consisted in a (somewhat long) appropriation of the basics of categorical semantic, of models of linear logic, of differential logic and of coherent differentiation. The second part consisted in showing that coherent differentiation is indeed a generalization of differential categories, a result that was somewhat expected but required some work.  Finally, the third part is too conceptual to explain for now and is hinted in \cref{sec:distributive-law}.

I feel that categorical semantic is a field that is particularly inaccessible and hard to vulgarize. The issue is not that this field is more complex or subtle than any other field. Rather, different notions tends to stack up really quickly. Besides, this field is ultimately based on category theory, a theory that is well known to be quite abstract and quite dividing between its adepts and its detractors. For this reason, there is little hope to explain my work in the scope of a report that is supposed to assume that the reader has no particular background in the field. So this report will mostly consists in an introduction. An introduction to the field of categorical semantic in \cref{sec:cat-semantic}. An introduction to how categorical semantic lead to the discovery of linear logic in \cref{sec:linear-logic}. An introduction to how categorical semantic lead to the discovery of differential linear logic in \cref{sec:diff-logic}. And finally, an introduction to coherent differentiation in \cref{sec:coherent-diff}. If I do my work correctly, then the reader should be able to understand what I worked on and for what reason I worked on those objects. Then I wrap up the report with \cref{sec:dill-is-coherent-diff} that gives a taste of the methodology I developed to show that coherent differentiation is a generalization of differentiation and with \cref{sec:distributive-law} as an appetizer on what I discovered in the last part of the internship.

\section{A (brief) introduction to categorical semantic} \label{sec:cat-semantic}

This introduction to categorical semantic is loosely inspired by this very enlightening paper from Melies \cite{Mellies_2009}. 

\subsection{A small and non exhaustive reminder on proof systems} \label{sec:proof-system}

I will assume that the reader is somewhat familiar with the notion of proof system, and stick to a minimal amount of material.

\begin{definition} \emph{Formulas} are inductively defined as the smallest set containing a set of propositional variables $V$ and which is closed under some set of connectives. For the sake of conciseness, there are only two connectives in this report: the conjunction ($\wedge$) and implication ($\imply$).
\[ F,G := v \in V \ | \ F \wedge G \ |  \ F \imply G \ | \ \top  \]
A \emph{context} is a list of formulas $\Gamma = A_1, \ldots, A_n$. A \emph{sequent} $\Gamma \vdash B$ consists in a context $\Gamma$ together with a formula $B$. It should be interpreted as the formula $A_1 \wedge A_2 \wedge \ldots \wedge A_n \ldots \imply B$. 
\end{definition}


\begin{definition} A \emph{proof system} consists in a set of \emph{inference rules} in the stylized form
\[ \frac{\Gamma_1 \vdash B_1 \ \ \ldots \ \ \Gamma_n \vdash B_n}{\Gamma \vdash B} \text{ \ \  and, for the case n=0  \ \ } \frac{}{\Gamma \vdash B} \]
The $\Gamma_i \vdash B_i$ are called the \emph{premises} of the rule, and $\Gamma \vdash B$ is called the \emph{conclusion}. An \emph{axiom} is a rule with no premises. 
\end{definition}

\begin{definition} A \emph{proof} of a sequent $\Gamma \vdash B$ consists in a labeled tree such that the label of the root is $\Gamma \vdash B$ and such that the label of each node is the conclusion of a rule whose premises are the labels of its children. The labels of the leafs should be the conclusion of an axiom.
\end{definition}

Usually, any sensible proof system contains two kinds of rules. First are structural rules, that allows to manipulate sequents. For example, the rules below are structural rules.
\[ \frac{}{A \vdash A} (\texttt{ax}) \ \ \frac{\Gamma \vdash B}{\Gamma, A \vdash B}(\texttt{weak}) \ \ \frac{\Gamma, A, A \vdash  B}{\Gamma, A \vdash B}(\texttt{contr}) \ \ \frac{\Gamma_1, A_1, A_2, \Gamma_2 \vdash B}{\Gamma_1, A_2, A_1, \Gamma_2 \vdash B}(\texttt{ex})  \]
The rule \texttt{ax} is the mandatory axiom rule. The weakening rules \texttt{weak} discard hypothesis. The contraction rules \texttt{contr} duplicate hypothesis. Finally, the exchange rule \texttt{ex} swap the position of formulas. The usage of this last rule is usually kept implicit.

Second, there are connectors rules that manipulate the different connectors. In this report, I stick to rules where the connector is only in the conclusion. The connector rules are classified depending on the position of the connector with respect to the $\vdash$ symbol: right or left.
\begin{center}
\begin{tabular}{l c c c}
  & Left rules & Right rules \\[1em]
$\wedge$: & $\dfrac{\Gamma, A_1 \vdash B}{\Gamma,  A_1 \wedge A_2 \vdash B} \ \dfrac{\Gamma, A_2 \vdash B}{\Gamma,  A_1 \wedge A_2 \vdash B}$ & $\dfrac{\Gamma \vdash B_1 \ \ \Gamma \vdash B_2}{\Gamma \vdash B_1 \wedge B_2}$ & \text{($\wedge$-rules)} \\[2em]
$\top$ & No rule & $\dfrac{}{\Gamma \vdash \top}$ & \text{($\top$-rules)} \\[2em]
$\imply$:  & $\dfrac{\Gamma_1 \vdash A \ \ \ \ \ B, \Gamma_2 \vdash C}{A \imply B , \Gamma_1, \Gamma_2 \vdash C}$ & $\dfrac{\Gamma, A \vdash B}{\Gamma \vdash A \imply B}$ & \text{($\imply$-rules)}
\end{tabular}
\end{center}
Left rules give a way to \emph{use} an hypothesis that contains a connector. Right rules give a way to \emph{build} a conclusion that contains a connector. There is one last structural rule that is crucial for any proof system.

\begin{definition} The cut rule is the following structural rule: 
\begin{prooftree}
	\hypo{\Gamma_1 \vdash B} 
	\hypo{B, \Gamma_2 \vdash C}
	\infer2{\Gamma_1, \Gamma_2 \vdash C}
\end{prooftree}
\end{definition}

\begin{definition} This proof system is a fragment of a bigger proof system called \emph{LJ} (intuitionistic logic). So I will call this system LJ for the rest of the report.
\end{definition}

\begin{remark} The reader already versed into proof theory and/or linear logic might have noticed that this system has a mix of \emph{additive} and \emph{multiplicative} rules. They have been chosen this way to make the flow of the presentation as minimal and natural as possible.
\end{remark}

The cut rule is the cornerstone of any proof systems as it allows to introduce any arbitrary ``lemma'' $B$. It can also be seen as a kind of composition of a proof of $A \imply B$ with a proof of $B \imply C$. Intuitively, this rule is not necessary: it is always possible to get rid of lemmas by ``unfolding their proof on the fly''. The formal pendant of this idea is called the \emph{cut elimination}.

\begin{theorem}[Cut elimination] The cut rule is admissible in LJ. In other word: for any proof using the cut rule, there exists a proof of the same sequent that does not use the cut rule. Besides, the proof of this theorem gives an explicit procedure called \emph{the cut elimination procedure} that takes as input any proof using the cut rule and outputs a proof of the same sequent  without the cut rule.
\end{theorem}

This cut elimination procedure is of tremendous importance as it introduces a notion of computation in proof: computing a proof consists in applying the cut elimination procedure, and the result of this computation will be a cut-free proof (a proof with no cut rule). In fact, cut elimination is the proof theoric equivalent of the normalization in the lambda calculus throught the Curry Howard isomorphism. Besides, the the cut free proofs are the proof theoric equivalent of the $\beta$-normal terms.

\subsection{Why categories are natural viewpoints on semantic}

The cut elimination procedure is all about \emph{how} proofs compute. Denotational semantic is instead more interested in \emph{what} proofs compute. The goal of denotational semantic is to associate to any proof  $\pi$ some interpretation  $\semantic{\pi}$ that is agnostic to the computational content of $\pi$.


\begin{definition} \label{def:invariant-modular} The interpretation of a proof should be invariant under cut elimination: if a proof $\pi'$ is obtained from $\pi$ by cut elimination, then $\semantic{\pi} = \semantic{\pi'}$. Besides, the interpretation should be \emph{modular}: If $\pi$ is a proof of shape
\begin{center}
\begin{prooftree} 
	\hypo{\pi_1} \infer1{\Gamma_1 \vdash B_1}
	\hypo{\ldots}
	\hypo{\pi_n} \infer1{\Gamma_n \vdash B_n}
	\infer3{\Gamma \vdash B}
\end{prooftree}
\end{center}
Then  $\semantic{\pi}$ should be uniquely determined by the interpretations $\semantic{\pi_1}, \ldots, \semantic{\pi_n}$ (so even if the proofs $\pi_i$ change, if their interpretation stays the same, the interpretation of $\pi$ remains the same).
\end{definition}

From those reasonable assumptions, we can infer quite a lot about the structure in which the interpretations $\semantic{\pi}$ must live. Take those two proof tree.
\begin{center}
($\pi$):  \begin{prooftree}
	\hypo{} \infer1{A \vdash A}
\end{prooftree} \ \ \ \ \
($\pi'$): \begin{prooftree}
	\hypo{\pi_1} \infer1{A \vdash B}
	\hypo{\pi_2} \infer1{B \vdash C}
	\infer2{A \vdash C}
\end{prooftree}
\end{center}
The proof $\pi$ suggest that for all formula $A$ there exists an entity $id_A := \semantic{\pi}$. The proof $\pi'$ and the constraint of modularity suggests that there exists
an operation $\circ$ on the entities (tied to some compatibility conditions) defined as $\semantic{\pi_2} \circ \semantic{\pi_1} := \semantic{\pi'}$. Besides, the structure of the cut elimination procedure together with the invariance condition implies some algebraic properties on $id_A$ and $\circ$ that are the ones of a \emph{category} (the full detailed process can be found in Mellies's paper~\cite{Mellies_2009}). So the semantic of a proof system \emph{has to be} a category. Denotational semantic can thus be renamed as \emph{categorical} semantic.

\begin{definition}[Models] \label{def:interpretation} A model of a proof theory consists of a category $\category$ together with
\begin{itemize}
 \item An interpretation of formulas: there exists a function $\semantic{.} : \{\text{formulas}\} \arrow Obj(\category)$. This interpretation extends to contexts, by defining $\semantic{A_1, \ldots, A_n} := \semantic{A_1 \wedge \ldots \wedge A_n}$. 
 \item An interpretation of proof: for any sequent $\Gamma \vdash B$,  there exists a function $\semantic{.}_{\Gamma,B} : \{\text{proofs of $\Gamma \vdash B$}\} \arrow \category(\semantic{\Gamma}, \semantic{B})$. Usually, $\Gamma$ and $B$ are kept implicit and I will write $\semantic{\pi}$ for $\semantic{\pi}_{\Gamma,B}$. This interpretation should be invariant and modular in the sense of \cref{def:invariant-modular}.
 \end{itemize}
\end{definition}

\subsection{Exhibiting the structure that $\category$ must have}\label{sec:models-of-logic}

Assume that $\category$ is a model of LJ. Let us infer the structure that $\category$ must have. First, 
the interpretation  of the proof \begin{prooftree}
	\hypo{} \infer1{A \vdash \top}
\end{prooftree} gives a morphism $\prodFinal_{\semantic{A}} : \semantic{A} \arrow \semantic{\top}$. Besides, a syntactical property of the cut elimination procedure called $\eta$-expansion together with invariance and modularity imply that any proof of $A \vdash \top$ must have for interpretation $\prodFinal_{\semantic{A}}$ (I will not go into further details). So $\category$ must have what is called a \emph{final object}.

\begin{definition} A final object in $\category$ is an object $\top_{\category}$ such that for any object $A$ of $\category$, there exists a unique morphism $\prodFinal_A :  A \arrow \top_{\category}$.
\end{definition}

Furthermore, take the three proofs below.
\begin{center}
$\pi:$ \begin{prooftree}
	\hypo{\pi_1} \infer1{A \vdash B_1}
	\hypo{\pi_2} \infer1{A \vdash B_2}
	\infer2{A \vdash B_1 \wedge B_2}
\end{prooftree} \ \ \
$\pi':$ \begin{prooftree}
	\hypo{} \infer1{B_1 \vdash B_1}
	\infer1{B_1 \wedge B_2 \vdash B_1}
\end{prooftree} \ \ \ 
$\pi'':$ \begin{prooftree}
	\hypo{} \infer1{B_2 \vdash B_2}
	\infer1{B_1 \wedge B_2 \vdash B_1}
\end{prooftree}
\end{center}
By modularity, the interpretation of $\pi$ is completely characterized by $\semantic{\pi _1}$ and $\semantic{\pi_2}$. So we can define an operator $\prodPair{.}{.}$ called the \emph{pairing} as $\prodPair{\semantic{\pi_1}}{\semantic{\pi_2}} := \semantic{\pi} : \semantic{A} \arrow \semantic{B_1 \wedge B_2}$. Besides, the interpretations of the proofs $\pi'$ and $\pi''$ provide two morphisms $\prodProj_1:=\semantic{\pi'} :  \semantic{B_1 \wedge B_2} \arrow \semantic{B_1}$  and $\prodProj_2 := \semantic{\pi''} :  \semantic{B_1 \wedge B_2} \arrow \semantic{B_2}$.
Again, invariance implies that $\prodProj_i \circ \prodPair{\semantic{\pi_1}}{\semantic{\pi_2}} = \semantic{\pi_i}$ and $\prodProj_2 \circ \prodPair{\semantic{\pi_1}}{\semantic{\pi_2}} = \semantic{\pi_2}$ and the $\eta$-expansion implies that such pairing should be unique.
As a result, a model of LJ must be what is called a Cartesian category.

\begin{definition} A category $\category$ is said to be Cartesian if it has a final object, and if for any pairs of objects $Y_1$ and $Y_2$, there exists an object $Y_1 \times Y_2$ and two morphisms $\prodProj_1 : Y_1 \times Y_2 \arrow Y_1$ and $\prodProj_2 : Y_1 \times Y_2 \arrow Y_2$ such that for any object $X$ and morphisms $f_1: X \arrow Y_1$, $f_2 : X \arrow Y_2$, there exists a unique morphism $\prodPair{f_1}{f_2} : X \arrow Y_1 \times Y_2$ such that the following diagram commutes.
\begin{center}
\begin{tikzcd}
    & X \arrow[ld, "f_1"'] \arrow[rd, "f_2"] \arrow[dd, "\prodPair{f_1}{f_2}" description, dashed] &     \\
Y_1 &                                                                                              & B_2 \\
    & Y_1 \times Y_2 \arrow[lu, "\prodProj_1"] \arrow[ru, "\prodProj_2"']                           &    
\end{tikzcd}
\end{center}
\end{definition}

The strength of categorical semantic is that it interprets complex syntactical interactions by this kind of simple property  that consists in the existence of an object and of a unique morphism following some commutations. Those properties are called \emph{Universal Mapping Property} (UMP). Any UMP generates a functor, such as the one below in the case of the Cartesian product.

\begin{definition} \label{def:closeness}
This property induces a \emph{bifunctor} (see the \href{sec:annex}{annex} for the terminology) $\_ \times \_ : \category \times \category \arrow \category $ that maps two objects $(X_1, X_2)$ to $X_1 \times X_2$ and two morphism $f_1 : X_1 \arrow Y_1$ and $f_2 : X_2 \arrow Y_2$ to a morphism $f_1 \times f_2 := \prodPair{f_1 \circ \prodProj_1}{f_2 \circ \prodProj_2} : X_1 \times X_2 \arrow Y_1 \times Y_2$ that consists in ``applying $f_1$ to the first coordinate and applying $f_2$ to the second''. 
\end{definition}

\begin{remark} \label{rem:associativity} This functor admits a natural isomorphism $\tensorAssoc^{\times}_{X,Y,Z} := \prodPair{\prodProj_1 \circ \prodProj_1}{\prodPair{\prodProj_2 \circ \prodProj_1}{\prodProj_2}}: (X \times Y) \times Z \arrow X \times (Y \times Z)$. It allows us to  write  $X_1 \times \ldots \times X_n$ and $f_1 \times \ldots \times f_n$  ``up to any arbitrary choice of parenthesis'' by keeping implicit the use of $\alpha^{\times}$. $\category$ is said to admit all finite products.
\end{remark}

The arrow has a similar interpretation. Because of length consideration, I will not do the whole process again, but directly introduce the universal mapping property required.

\begin{definition} Given a bifunctor $\_ \times \_$, a category $\category$ is said to be closed (with regard to $\times$) if for all pairs of object $X$ and $Y$, there exists an object $X \imply_{\category} Y$ and a morphism $\ev : (X \imply_{\category} Y) \times X \arrow Y$ such that for any object $Z$ and morphism $f: X \times Y \arrow Z$, there exists a unique morphism $\cur(f) : (X \imply_{\category} Z)$ such that the diagram below commutes.
\begin{center}
\begin{tikzcd}
X \times Y \arrow[r, "f"] \arrow[d, "\cur(f) \ \times \ id_Y"'] & Z\\
(Y \imply_{\category} Z) \times Y \arrow[ru, "\ev"']       &  
\end{tikzcd}
\end{center}
\end{definition}
For the reader not used to category theory, this definition might cause some issues. But any computer scientist used to functional programming an curryfication is in fact already used to this idea. Recall that curryfication means that any function with two arguments $f : A \times B \arrow C$ can be uniquely mapped to a function  $\cur(f) : A \arrow (B \arrow C)$ defined as $a \mapsto (b \mapsto f(a,b))$. Closedness is similar and means that $\category$ contains an object $B \imply_{\category} C$ that describes the ``set of morphism from $B$ to $C$'' and that any morphism $f : A \times B \arrow C$ can be uniquely seen as a morphism $\cur(f) : A \arrow (B \imply_{\category} C)$. In fact, Curryfication is a particular case of closeness where the morphism $\ev$ is the function that takes as input a function $h : A \arrow B$ and $x \in A$ and returns $h(x)$.

\begin{definition} A Cartesian closed category (CCC) is a category with a Cartesian product that is closed with regard to this product.
\end{definition}

We saw that the models of  LJ have to be Cartesian closed categories. Conversely, any Cartesian closed category $\category$ together with a function $\semantic{.} : X \arrow Obj(\category)$  defines a model. Indeed, the interpretation of formulas can be defined inductively as $\semantic{x} := \semantic{x}$, $\semantic{A \wedge B} := \semantic{A} \times \semantic{B}$, $\semantic{\top} := \top_{\category}$ and $\semantic{A \imply B} := \semantic{A} \imply_{\category} \semantic{B}$.  This interpretation extends to context by defining $\semantic{A_1, \ldots, A_n} = \semantic{A_1 \wedge \ldots \wedge A_n} = \semantic{A_1} \times \ldots \times \semantic{A_n}$ (using \cref{rem:associativity} that ensures parenthesis irrelevance). 

The interpretation of a proof is then defined inductively.. I will not give the full details, just that each of the rule is interpreted by its semantic equivalent: the rule (\texttt{ax}) is interpreted by an identity, the rule (\texttt{cut}) by a composition, the rules ($\wedge$-left) by the projections, the rule ($\wedge$-right) by a pairing, the rule ($\imply$-left) by $\ev$ and the ($\imply$-right) rule by a $\cur$.  It only leaves the rules (\texttt{contr}), (\texttt{weak}) and (\texttt{ex}). The rule (\texttt{weak}) can directly be interpreted by the first projection $\prodProj_1$. The rule (\texttt{contr}) is tied to the existence of a natural transformation $\comonoidSum_X^{\times} : X \arrow X \times X$ that can be defined as $\comonoidSum^{\times} := \prodPair{id}{id}$. . Finally, the exchange rule is tied to the existence of a natural transformation $\tensorSym^{\times}_{X,Y}: X \times Y \arrow Y \times X$ than we can be defined as $\tensorSym := \prodPair{\prodProj_2}{\prodProj_1}$.

\section{An introduction to linear logic} \label{sec:linear-logic}

Linear logic is a proof theory discovered by Girard \cite{Girard_1987} while he was studying a model $\coh$ (the category of coherent spaces) that he discovered \cite{Girard_1986}\footnote{There is a bit of storytelling here. The real story turns out to be even more interesting !}. Girard noticed that the object of the closure $A \imply B$ could be decomposed in two successive constructions $!A \linarrow B$ where $! : Obj(\coh) \arrow Obj(\coh)$ and $\_ \linarrow  \_ : Obj(\coh)^2 \arrow Obj(\coh)$.  Therefore, Girard introduced a category $\lin$ whose objects are the objects of $\coh$ and whose morphisms are defined as $\lin(A, B) := A \linarrow B$. This category admits a functor $! : \lin \arrow \lin$ and by design, for any objects $A,B$, $\coh(A, B) = \lin(!A, B)$. Besides, $\lin$ has some though inducing properties:
\begin{itemize}
	\item The object $A \times B$ still gives a Cartesian product in $\lin$ and $\top$ is still a final object. So $\lin$ is Cartesian.
	\item However, $\lin$ is not closed with regard to this product.
	\item In fact, $\lin$ is closed with regard to some bifunctor $\tensor$. Besides, there exists an object $1$ such that $(\lin, 1, \tensor)$ is what is called a ``symetric monoidal category'', see \cref{def:sm-category} below.
\end{itemize}

It turns out that there is a strong analogy between these objects and linear algebra. Let $\vect$ be the category whose objects are the $\R$-vector spaces and whose morphisms are the (continuous) linear maps. $\vect$ is a Cartesian category: the final object $\top$ is the trivial vector space \{0\} and the Cartesian product is nothing more than the usual product of two vector spaces $E \times F$. Besides, this category seems to be closed: $\vect(E,F)$ is a vector space itself. However it is not closed with regard to the cartesian product: giving a function $f \in \vect(E, \vect(F,G))$ is the same as giving a \emph{bilinear} application $f : E \times F \arrow G$, not a linear map.

 However it is possible to define a vector space $E \tensor F$ called the ``tensor product of E and F''. I will only give its definition in finite dimension, as it is more intuitive. If $E$ is a vector space of dimension $n$ and basis $(e_i)$ and if $F$ is a vector space of dimension $m$ and basis $(f_j)$, then $E \tensor F$ is a vector space of dimension $n.m$ with basis denoted as $(e_i \tensor f_j)$. The idea of this space is to factorize bilinear applications: any bilinear application $\phi : E \times F \arrow G$ can be characterized as a (unique) linear application $\overline{\phi} : E \tensor F \arrow G$ defined on the basis as $\overline{\phi}(e_i \tensor f_j) := \phi(e_i,f_j)$. So $\vect(E, \vect(F,G))$ coincides with the set of bilinear applications $E \times F \arrow G$ that coincides with $\vect(E \tensor F, G)$. This is exactly a closure with regard to $\tensor$.

\subsection{Symetric monoidal categories} \label{sec:mall}

In the two examples above, the functor $\tensor$ gives to the category the structure of a \emph{symetric monoidal category}.  This notion existed prior to linear logic, see \cite{Lane_1963}. The definition below is a refinement due to Kelly~\cite{Kelly_1964}.

\begin{definition}  \label{def:sm-category} A \emph{monoidal} category $(\category, \tensor, 1)$ consists in a category $\category$, a bifunctor $\tensor$ on $\category$ and an object $1$ of $\category$ called the unit such that there exists three natural isomorphism $\tensorAssoc_{A,B,C} : (A \tensor B) \tensor C \arrow A \tensor (B \tensor C)$,  $\tensorUnitL_A : 1 \tensor A \arrow A$, $\tensorUnitR_A : A \tensor 1 \arrow A$ that folows the commutations below..
\begin{center}
\begin{tikzcd}
(A \tensor 1) \tensor B \arrow[r, "\tensorAssoc"] \arrow[rd, "\tensorUnitL \tensor id_B"'] & A \tensor (1 \tensor B) \arrow[d, "id_A \tensor \tensorUnitR"] \\
                                                                                           & A \tensor B                                                            
\end{tikzcd}
\begin{tikzcd}[column sep = small]
                                                                                    & (A \tensor B) \tensor (C \tensor D) \arrow[ld, "{\tensorAssoc_{A, B, C \tensor D}}"'] &                                                                                                                                        \\
A \tensor (B \tensor (C \tensor D))                                                 &                                                                                       & ((A \tensor B) \tensor C) \tensor D \arrow[lu, "{\tensorAssoc_{A \tensor B, C, D}}"'] \arrow[d, "{\tensorAssoc_{A,B,C} \tensor id_D}"] \\
A \tensor ((B \tensor C) \tensor D \arrow[u, "{id_A \tensor \tensorAssoc_{B,C,D}}"] &                                                                                       & (A \tensor (B \tensor C)) \tensor D \arrow[ll, "{\tensorAssoc_{A, B \tensor C, D}}"]                                                  
\end{tikzcd}
\end{center}
A \emph{symetric monoidal} category $\category$ is a monoidal category that is also equiped with a natural transformation $\tensorSym_{A,B} : A \tensor B \arrow B \tensor A$. $\tensorSym$ is required to be an involution: $\tensorSym_{B,A} \circ \tensorSym_{A,B} = id_{A \tensor B}$. Finally, it should follw those commutations below.
\begin{center}
\begin{tikzcd}
1 \tensor A \arrow[r, "{\tensorSym_{1,A}}"] \arrow[rd, "\tensorUnitL"'] & A \tensor 1 \arrow[d, "\tensorUnitR"] \\
                                                                        & A                                    
\end{tikzcd}
\begin{tikzcd}
(A \tensor B) \tensor C \arrow[r, "{\tensorAssoc_{A,B,C}}"] \arrow[d, "{\tensorSym_{A,B} \tensor id_C}"'] & A \tensor (B \tensor C) \arrow[r, "{\tensorSym_{A, B \tensor C}}"]    & (B \tensor C) \tensor A \arrow[d, "{\tensorAssoc_{B,C,A}}"] \\
(B \tensor A) \tensor C \arrow[r, "{\tensorAssoc_{B,A,C}}"']                                              & B \tensor (A \tensor C) \arrow[r, "{id_B \tensor \tensorSym_{A,C}}"'] & B \tensor (C \tensor A)                                    
\end{tikzcd}
\end{center}
\end{definition}

The axioms of monoidal categories are the minimal axioms required to ensure that we can rewrite parenthesis, collapse the unit, and swap formulas in any arbitrary order and still end up with the same morphism (it is a generalization of the regular notion of monoid). Associativity in particular allows to define $A_1 \tensor \ldots \tensor A_n$ and $f_1 \tensor \ldots \tensor f_n$ ``up to parenthesis rewriting'', in a similar way as \cref{rem:associativity}. In fact, a Cartesian category is a particular kind of symetric monoidal category.

\begin{remark} \label{rem:product-is-sm} If $\category$ is Cartesian, the functor $\_ \times \_$  associated to a Cartesian product gives to $\category$ the structure of a symetric monoidal category $(\times, \top, \tensorAssoc^{\times}, \tensorUnitL^{\times}, \tensorUnitR^{\times})$, where $\tensorAssoc^{\times}$ and $\tensorSym^{\times}$ are defined in \cref{sec:models-of-logic}, $\tensorUnitL^{\times}_A = \prodProj_1 : A \times \top \arrow A$, $\tensorUnitR^{\times}_A = \prodProj_2 : \top \times A \arrow A$. Note that $\tensorUnitL^{\times}$ and $\tensorUnitR^{\times}$ are isomorphisms because $\top$ is a final object (their inverses are respectively $\prodPair{id_A}{\prodFinal_A}$ and $\prodPair{\prodFinal_A}{id_1}$). So the notion of symetric monoidal category was already the notion we implicitely used when we said that ``parenthesis for $\times$ are irrelevant and we can interpret the exchange rule''.
\end{remark}



\subsection{Deducing the syntax from the semantic}

These models suggested the existence of a logic where two notions of conjunction might coexist. So Girard introduced a logical system in which formulas had three binary connectives: one implication $\linarrow$ (called the linear arrow), and two conjunction $\tensor$ (called the tensor) and $\with$ (called the with). The semantic of $\tensor$ should be $\tensor$, the semantic of $\linarrow$ should be $\linarrow$, and the semantic of $\with$ should be $\times$. There are also two constant symbols $1$ (associated to $\tensor$) and $\top$ (associated to $\with$) whose semantic should be respectively the objects $1$ and $\top$. Girard introduced the following syntax.

\begin{center}
\begin{tabular}{l c c c}
  & Left rules & Right rules \\[1em]
$\with$: & $\dfrac{\Gamma, A_1 \vdash B}{\Gamma,  A_1 \with A_2 \vdash B} \ \dfrac{\Gamma, A_2 \vdash B}{\Gamma,  A_1 \with A_2 \vdash B}$ & $\dfrac{\Gamma \vdash B_1 \ \ \Gamma \vdash B_2}{\Gamma \vdash B_1 \with B_2}$ & \text{($\with$-rules)} \\[2em]

$\top$ & No rule & $\dfrac{}{\Gamma \vdash \top}$ & \text{($\top$-rules)} \\[2em]

$\linarrow$:  & $\dfrac{\Gamma_1 \vdash A \ \ \Gamma_2, B \vdash C}{\Gamma_1, \Gamma_2, A \linarrow B \vdash C}$ & $\dfrac{\Gamma, A \vdash B}{\Gamma \vdash A \linarrow B}$ & \text{($\linarrow$-rules)} \\[2em]

$\tensor$ & $\dfrac{\Gamma, A_1, A_2 \vdash B}{\Gamma,  A_1 \tensor A_2 \vdash B}$ & $\dfrac{\Gamma_1 \vdash B_1 \ \ \Gamma_2 \vdash B_2}{\Gamma_1, \Gamma_2 \vdash B_1 \tensor B_2}$ & \text{($\tensor$-rules)}  \\[2em]

$1$ & $\dfrac{\Gamma \vdash B}{\Gamma, 1 \vdash B}$ & $\dfrac{}{\vdash 1}$ & \text{($1$-rules)}
\end{tabular}
\end{center}
Together with the structural rules.
\[ \frac{}{A \vdash A} (\texttt{ax}) \ \ \ \ \frac{\Gamma_1, A_1, A_2, \Gamma_2 \vdash B}{\Gamma_1, A_2, A_1, \Gamma_2 \vdash B}(\texttt{ex})  \ \ \ \ \frac{\Gamma_1 \vdash B \ \ \Gamma_2, B \vdash C}{\Gamma_1, \Gamma_2 \vdash C} (\texttt{cut}) \]

Let us motivate those rules with semantic in mind. The rules associated to $\with$, $\top$ and $\linarrow$ are the same as the rule that were associated to $\wedge$, $\top$  and $\imply$, because as we saw the semantic counterpart of those rules is the structure of a Cartesian product, of a final object, and of a closure.  Besides, a sequent $A_1, \ldots, A_n \vdash B$ should be interpreted as the formula $A_1 \tensor \ldots \tensor A_n \linarrow B$ (and not $A_1 \with \ldots \with A_n \linarrow B$)  as the linear arrow is closed with regard to $\tensor$ (and not $\times$). In particular, the ($\tensor$-left) rule is  nothing more than some parenthesis rewriting. On the other hand, the ($\tensor$-right) rule betrays the fact that $\tensor$ is a bifunctor: given two morphisms $f_1 : \semantic{\Gamma_1} \arrow \semantic{B_1}$ and $f_2 : \semantic{\Gamma_2} \arrow \semantic{B_2}$, there exists a morphism $f_1 \tensor f_2 : \semantic{\Gamma_1} \tensor \semantic{\Gamma_2} \arrow \semantic{B_1} \tensor \semantic{B_2}$. Similarly, the right (1-rule) does not say anything semantically (the empty context is interpreted as the formula 1), while the left (1-rule) is the syntactic equivalent of the morphism $\tensorUnitR_{\semantic{\Gamma}} : \semantic{\Gamma} \tensor 1 \arrow \semantic{\Gamma}$. Finally, the exchange rule carries to linear logic thanks to the symmetry $\tensorSym_{\semantic{A}, \semantic{B}} : \semantic{A} \tensor \semantic{B} \arrow \semantic{B} \tensor \semantic{A}$.

\begin{definition} \label{def:mall} The proof system introduced in this section is called iMALL (Intuitionistic Multiplicative Additive Linear Logic). A model of iMALL is a symetric monoidal category $(\category, \tensor, 1, \tensorAssoc, \tensorUnitL, \tensorUnitR, \tensorSym)$ closed with regard to $\tensor$ that also contains a Cartesian product $\times$. From this point onward, I will write the product $\with$ instead of $\times$ in order to match with the logical connector.
\end{definition}

Linear logic should be seen as a logic about \emph{resources management}. A formula $A \linarrow B$ states that ``using the resource  A once and only once produces the outcome B''. The formula $A \tensor B$ is interpreted as ``having both a resource  A and a resource B at the same time''. So $A, A \vdash A \tensor A$ is provable, but not $A \vdash A \tensor A$. The formula $A \with B$ is interpreted as ``choosing at will between a resource A and a resource B, but not being able to chose both''   ''. So $A \vdash A \with A$ is provable, but not $A,A \vdash A \with A$.

\subsection{About the weakening and contraction rules} \label{sec:weakening-and-contraction}

The most notable thing about this new system is the absence of weakening and contraction rules. 
There is an obvious syntactical reason. The presence of those two rules would imply the provability of the sequents $A \tensor B \vdash A \with B$, $A \with B \vdash A \tensor B$, $1 \vdash \top$  and $\top \vdash 1$. It would lead to a collapse between the connector $\tensor$ with the connector $\with$. This is exactly what happens in the system LJ that I introduced in \cref{sec:proof-system}.

 There is a similar observation in the world of semantic. Assume that the weakening and contraction rules hold \emph{for a given formula $A$}. Take the proofs bellow.
\begin{center}
\begin{prooftree}
	\hypo{} \infer1{\vdash 1}
	\infer1{A \vdash 1}
\end{prooftree} \ \ \ \ \ 
\begin{prooftree}
	\hypo{} \infer1{A \vdash A}
	\hypo{} \infer1{A \vdash A}
	\infer2{A,A \vdash A \tensor A}
	\infer1{A \vdash A \tensor A}
\end{prooftree}
\end{center}
The interpretation of the left proof should be a morphism $\comonoidUnit_{\semantic{A}} : \semantic{A} \arrow 1$. The interpretation of the right proof should be a morphism $\comonoidSum_{\semantic{A}} : \semantic{A} \arrow \semantic{A} \tensor \semantic{A}$. So if the rules hold for every formula,  the tensor $\tensor$ would turn into a Cartesian product accordingly to the fact below.

\begin{fact} \label{thm:sm-is-cartesian}
A symetric monoidal category is Cartesian if and only if  there exists two natural transformations $\comonoidSum_X : X \arrow X \tensor X$ and $\comonoidUnit_X : X \arrow 1$ that follow some commutations that I will not detail here (see prop.16 in section 6.3 of \cite{Mellies_2009}  for the complete statement)
\end{fact}

\begin{proof}[Sketch of the proof] For the direct implication, take $\comonoidUnit_A := \prodFinal_A$ (recall that $\top$ is a final object) and $\comonoidSum_A := \prodPair{id_A}{id_A}$ (note that I introduced this morphism in \ref{sec:models-of-logic} for interpreting the contraction rule). For the reverse implication, the pairing and the projections are defined as follow.  
\begin{center}
\begin{tikzcd}
\prodPair{f_1}{f_2}: X \arrow[r, "\Delta_X"] & X \tensor X \arrow[r, "f_1 \tensor f_2"] & Y_1 \tensor Y_2
\end{tikzcd}
\end{center}
\begin{center}
\begin{tikzcd}
\prodProj_1: X_1 \tensor X_2 \arrow[r, "id_{X_1} \tensor \prodFinal_{X_2}"] & X_1 \tensor 1 \arrow[r, "\tensorUnitL"] & X_1 
\end{tikzcd} and 
\begin{tikzcd}
\prodProj_2: X_1 \tensor X_2 \arrow[r, "\prodFinal_{X_1} \tensor id_{X_2}"] & 1 \tensor X_2 \arrow[r, "\tensorUnitR"] & X_2
\end{tikzcd}
\end{center}
\end{proof}

However, linear logic is still supposed to be a refinement of classical logic (recall that $\lin(!A, B) = \coh(A,B)$). So Girard introduced a new unary connector $!$ called \emph{the exponential} (that would be the syntactical equivalent of the functor $!$ in $\lin$) and constrained the weakening and contraction to formulas of shape $!A$. So there are four new rules, a right rule for $!$ (called promotion), a  left rule for $!$ (called dereliction), a weakening and a contraction.

\begin{center}
\begin{prooftree}
	\hypo{\Gamma, A \vdash B} \infer1{\Gamma, !A \vdash B}
\end{prooftree} (\texttt{der}) \ \ \ 
\begin{prooftree}
	\hypo{!A_1, \ldots, !A_n \vdash B} \infer1{!A_1, \ldots, !A_n \vdash !B}
\end{prooftree}(\texttt{prom}) \ \ \ 
\begin{prooftree}
	\hypo{\Gamma \vdash B} \infer1{\Gamma, !A \vdash B}
\end{prooftree}(\texttt{weak}) \ \ \ 
\begin{prooftree}
	\hypo{\Gamma, !A, !A \vdash B} \infer1{\Gamma, !A \vdash B}
\end{prooftree} (\texttt{contr})
\end{center}

One should see a formula $!A$ as a resource that can be ``duplicated' and discarded''. The rule (\texttt{weak}) states that a ressourcce $!A$ can be discarded. The rule (\texttt{contr}) states that a ressource $!A$ can be duplicated. The rule (\texttt{der}) just ``forget'' that $A$ was used once and only once. Finally, the rule (\texttt{prom}) states that if all of the resources used to produce $B$ are discardable and duplicable, then $B$ can be produced an arbitrary number of time (including zero) by discarding and duplicating the proof.

Besides, a reasoning similar to the one at the begining of this section shows that the contraction and weakening rules implies the existence of two natural transformations $\weak_X : !X \arrow 1$ and $\contr_X : \ !X \arrow !X \tensor !X$. Conversely, those two morphisms give an interpretation for the proofs that end with a weakening or a contraction rule.

\subsection{The Seely isomorphisms} \label{sec:seely}

There are still many concurrent notions of models of linear logic, and chosing one is usually a matter of personnal taste. I chose the notion of  \emph{Seely categories} introduced in \cite{Seely_1989}. This is not the most general one, but it is quite powerfull and it arises in most of the concrete models. The starting point of this axiomatization is the observation made by Girard and deepened by Seely that in syntax as well as in many models, the connective $!$ ``transports'' the connectives $\with$ to the connectors $\tensor$. 

\begin{definition} \label{def:seely-isomorphisms}  A \emph{Seely category} is a model in which the functor $!$ is a \emph{strong symetric monoidal functor} from $(\category, \with, \top)$ to  $(\category, \tensor, 1)$. It means that there exists an isomorphism $\seelyOne : 1\simeq !\top$ and a natural isomorphism $\seelyTwo_{A,B} : \  !A \tensor !B \simeq !(A \with B)$\footnote{Recall that $\with$ is the new name for $\times$, as stated in \cref{def:mall}} called the \emph{Seely isomorphisms} that follow the commutations below. 
\begin{center}
\begin{tikzcd}
(!A \tensor !B) \tensor !C \arrow[r, "\seelyTwo \tensor id"] \arrow[d, "\tensorAssoc"'] & !(A \with B) \tensor !C \arrow[r, "\seelyTwo"] & !((A \with B) \with C) \arrow[d, "! \tensorAssoc^{\with}"] \\
!A \tensor (!B \tensor !C) \arrow[r, "id \tensor \seelyTwo"']                           & !A \tensor !(B \with C) \arrow[r, "\seelyTwo"']      & !(A \with (B \with C))                                     
\end{tikzcd} \ \ 
\begin{tikzcd}
!A \tensor !B \arrow[d, "\tensorSym"'] \arrow[r, "\seelyTwo"] & !(A \with B)                                     \\
!B \tensor !A \arrow[r, "\seelyTwo"']                         & !(B  \with A) \arrow[u, "!\tensorSym^{\with}"']
\end{tikzcd}
\begin{tikzcd}
!A \tensor 1 \arrow[r, "\tensorUnitR"] \arrow[d, "id \tensor \seelyOne"'] & !A                                                    \\
!A \tensor !\top \arrow[r, "\seelyTwo"']                                  & !(A \with \top) \arrow[u, "!\tensorUnitR^{\with}"']
\end{tikzcd}
\begin{tikzcd}
1 \tensor !A \arrow[r, "\tensorUnitL"] \arrow[d, "\seelyOne \tensor id"'] & !A                                                    \\
! \top \tensor !A \arrow[r, "\seelyTwo"']                                 & !(\top \with A) \arrow[u, "!\tensorUnitL^{\with}"']
\end{tikzcd}
\end{center}
Basically, those diagrams states that doing symetric monoidal reasoning on $!(\_ \with \_)$ using the symetric monoidal structure of $\with$ below $!$ is the same as doing symetric monoidal reasoning on $!\_ \tensor !\_$ using the symetric monoidal structure of $\tensor$.
\end{definition}

\begin{remark} \label{rem:lax-smf} Here, $!$ is a strong monoidal functor from $(\category, \tensor, 1)$ to $(\category, \with, \top)$ but this notion obviously extend to any functor $F$ that goes from a symetric monoidal category $(\category, \tensor, 1)$ to another symmetric monoidal category $(\category', \tensor', 1')$. Besides, when the two natural transformations are not isomorphisms,  $F$ is called a \emph{lax symetric monoidal functor}. We will come accross this notion latter.
\end{remark}

\begin{remark} \label{rem:comonoid-from-seely} Any Seely category contains the two natural trasnformation $\weak_A : !A \arrow 1$ and $\contr_A : !A \arrow !A \tensor !A$. They are obtained by ``lifting'' the two morphisms $\Delta^{\with}_A := \prodPair{id_A}{id_A} : A \arrow A \with A$ and $\epsilon^{\with}_A := \prodFinal_A : A \arrow \top$ that characterize $\with$ as a Cartesian product (recall \cref{thm:sm-is-cartesian}) throught the functor $!$ and carry them to $\tensor$ with the Seely isomorphisms.
\begin{center}
$\contr_A$: \begin{tikzcd}
!A \arrow[r, "! \comonoidSum^{\with}"] & !(A \with A) \arrow[r, "(\seelyTwo)^{-1}"] & !A \tensor !A
\end{tikzcd} \ \ \ \ \ 
$\weak_A$: \begin{tikzcd}
!A  \arrow[r, "! \comonoidUnit^{\with}"] & !\top \arrow[r, "(\seelyOne)^{-1}"] & 1
\end{tikzcd}
\end{center}
So Seely categories does not introduce those two morphisms directly as axioms, but derives them from other objects. This is an usual dynamic of the interraction between syntax and semantic: the axiomatization of the semantic is usually obtained by a direct correspondance with the syntactical rules, but sometimes it is better to divert from it a bit.
\end{remark}

\subsection{The structure of $!$}

Fortunately, the structure that $!$ should have is  more consensual. It should be a comonad, commonly called \emph{the exponential comonad}.

\begin{definition} A \emph{comonad} on $\category$ consists in $(!, \der, \dig)$ where $!\_ : \category \arrow \category$ is a functor together with two natural transformations $\der_A : !A \arrow A$ and $\dig_A : !A \arrow !!A$ such that the following diagrams commute.
\begin{center}
\begin{tikzcd}
!A \arrow[d, "\dig_A"'] \arrow[r, "\dig_A"] & !!A \arrow[d, "\dig_{!A}"] \\
!!A \arrow[r, "!\dig_A"']                   & !!!A                      
\end{tikzcd} \ \ \ \ 
\begin{tikzcd}
   & !A \arrow[ld, equal, "id_{!A}"'] \arrow[d, "\dig_A" description] \arrow[rd, equal, "id_{!A}"]    &    \\
!A & !!A \arrow[l, "\der_{!A}"] \arrow[r, "!\der_A"] & !A
\end{tikzcd}
\end{center}
The left diagram is called the Monadic square and the right diagram is called the monadic triangle.
\end{definition}

\begin{definition}[Dereliction] \label{def:prom} For any morphism $f : A \arrow B$, the morphism $\kleisliCastExp(f):$ 
\begin{tikzcd}
!A \arrow[r, "\der_A"] & A \arrow[r, "f"] & B
\end{tikzcd} is called \emph{the dereliction} of $f$. It gives an interpretation for the dereliction rule.
\end{definition}

\begin{definition}[Promotion] \label{def:der}
For any morphism $f : \ !A \arrow B$ the morphism 
$f^! :$~\begin{tikzcd}
!A \arrow[r, "\dig_{A}"] & !!A \arrow[r, "!f"] & !B
\end{tikzcd} is called the \emph{promotion} of $f$. It gives an interpretation for the promotion rule when $n =1$. There is a generalization of this construction for any $n$ that can be derived from the Seely isomorphisms but I will not give the details here.
\end{definition}

Putting everything together ends up with the following notion.
\begin{definition}[Models of LL]
The models of linear logic are \emph{symetric monoidal} categories $(\category, \tensor, A)$, \emph{closed} with regard to $\tensor$, that contains a \emph{Cartesian product} $(\with, \top)$, that contains a \emph{comonad} $(!, \der, \dig)$ and two natural isomorphisms  $\seelyOne : 1\simeq !\top$ and $\seelyTwo_{A,B} : \  !A \tensor !B \simeq !(A \with B)$ such that $(!, \seelyOne, \seelyTwo)$ is a \emph{srong symetric monoidal functor}. 
\end{definition}

Besides, a property similar to $\lin(!A, B) = \coh(A,B)$ holds in any model of linear logic. Indeed, the category $\coh$ is a particular instance of the generic notion of \emph{Kleisli category}.

\begin{definition} If $(!, \der, \dig)$ is a comonad on $\category$, the Kleisli category of $!$ is a category $\kleisliExp$ with objects $Obj(\kleisliExp) = Obj(\category)$ and morphisms $\kleisliExp(A, B) = \category(!A, B)$. The identity is defined from the dereliction as the morphism $\der_A \in  \category(!A, A) =  \kleisliExp(A, A)$. The composition is defined from the promotion: for any $f \in \kleisliExp(A, B) = \category(!A, B)$ and $g \in \kleisliExp(B, C) = \category(!B, C)$, $g \kleisliExpComp f := g \circ f^!$.
\end{definition}

\begin{proof} The proof that it is a category is left for the reader. The neutrality of the identity is a consequence of the triangle equality of the comonad. The associativity is a consequence of the square of the coMonad.
\end{proof}

If the category $\category$ can be seen as the category of the ``linear world'', the category $\kleisliExp$ on the other hand can be seen as an extension of $\category$ that introduces a way to produce ``non linearity'' (by allowing to duplicate the resources). The \emph{dereliction} introduced in \cref{def:der} allows to ``forget'' that a given morphism is linear.

\begin{proposition} The category  $\kleisliExp$ \emph{extends} $\category$ in the sense that the dereliction $\kleisliCastExp(f)$ induces a functor $\kleisliCastExp : \category \arrow \kleisliExp$ by defining $\kleisliCastExp(A) := A$. Besides, this functor is \emph{faithful}, meaning that if $\kleisliCastExp(f) = \kleisliCastExp(g)$, then $f=g$.
\end{proposition}

\begin{theorem} The Kleisli category $\kleisliExp$ is cartesian closed, and thus is a model of LJ.
\end{theorem}

\begin{proof} I will not detail the proof here, but I will give the important ideas. The final object is still $\top$ and the Cartesian product of $A$ with $B$ is still the object $A \with B$. The projections are defined as $\kleisliCastExp(\prodProj_1) : !(A \with B) \arrow A$ and $\kleisliCastExp(\prodProj_2) : !(A \with B) \arrow B$.  The closure then comes from the fact that we can transport the closure in $\category$ with regard to $\tensor$ to a closure in $\kleisli$ with regard to $\with$ thanks to the Seely isomorphism $\seelyTwo_{X,Y}: ! X \tensor !Y \arrow !(X \with Y)$. 
\end{proof}

\begin{remark} Actually, models of linear logic require one more compatibility condition between $\seelyTwo$ and $\dig$ that is necessary in the proof of the closure, but I will not talk about it here.
\end{remark}

\section{Differentiation in some models of linear logic}  \label{sec:diff-logic}
 
As hinted before with $\vect$, the ideas of linear logic have deep ties with linear algebra. However there is no satisfactory way of defining $(!, \der, \dig)$ on $\vect$, so interpreting the exponential through the lenses of linear algebra seems at first to be a lost cause. Fortunately, Ehrhard discovered a model of LL called $\fin$ (\emph{finiteness spaces}) \cite{Ehrhard_2005} in which the objects are (topological) vector spaces and the arrows are the (continuous) linear maps. In this model, the exponential can be seen as a power series construction: $\fin(!X, Y)$ is the set of \emph{analytic} functions from X to Y. For example, a function $f : \R \arrow \R$ is analytic if $f(x) = \Sigma_{n=0}^{\infty} a_n x^n$. 

Analytic functions are a particular case of smooth functions. Recall the fundamental idea of differential calculus: for any $x \in X$, the variation a smooth function $f : X \arrow Y$ around $x$ can be approximated  as a linear variation $f(x+u) \approx f(x) + f'(x).u$, where $f'(x)$ is a  (continous) linear map called the differential of $f$ in $x$. So there exists a function $f': X \arrow \fin(X,Y)$.
Besides, one can show that when $f$ is analytic $f'$ is also analytic, meaning that $f' \in \fin(!X, \fin(X,Y))$. Because of the closure with regard to $\tensor$, it means that the differential of $f$ can be seen as a morphism $f' \in \fin(!X \tensor X, Y)$. So the fundamental idea of differential calculus seems to corresponds to some categorical properties on $\fin$. Here are those properties, defined in a generic way.
 
 \begin{definition} A model of LL $\category$ is said to be \emph{additive} if for all objects $X, Y$, $\category(X,Y)$ is a monoid. In other word, if there exists a morphism $0 \in \category(X,Y)$ and an associative operator $+$ such that for all $f \in \category(X,Y)$, $0 + f = f + 0 = f$. Besides, the addition should be \emph{compatible} with the composition and the tensor, meaning: \begin{itemize}
	\item $g \circ (f_1 + f_2) \circ h = g \circ f_1 \circ h + g \circ f_2 \circ h$. In other words, the morphisms are ``linear'' in the sum
	\item $(f_1 + f_2) \tensor (g_1 \tensor g_2) = f_1 \tensor g_1 + f_2 \tensor g_1 + f_1 \tensor g_2 + f_2 \tensor g_2$. In other word, the tensor is ``bilinear'' in the sum.
\end{itemize}
\end{definition} 

\begin{definition} \label{def:diff-axioms} A \emph{differential category}  is an additive category $\category$ such that there exists a natural transformation $\dillDerive : !X \tensor X \arrow !X$. This natural transformation allows to define for any $f : !X \arrow Y$ a morphism $f' := f \circ \dillDerive : \  !X \tensor X \arrow Y$. Besides, $\dillDerive$ is required to follow some commutations that consist in separate interactions between $\dillDerive$ and $\weak$/$\contr$/$\der$/$\dig$/$\dillDerive$. I  will not give the diagrams, but explain to what they correspond in finiteness spaces \begin{itemize}
\item Interaction with $\der$: The derivation of a linear function is the linear function itself
\item Interaction with $\dig$: The chain rule holds, that is $(f \circ g)'(x) = f'(g(x)) \circ g'(x)$
\item Interaction with $\weak$: The derivative of a constant function is 0
\item Interaction with $\contr$: The Leibniz rule holds, that is if $\phi$ is a bilinear application, $\phi(f,g)'(x).u = \phi(f'(x).u, g(x)) + \phi(f(x), g'(x).u)$.
\item Interaction with $\dillDerive$: The Schwarz rule, that is $f''(x).(u,v) = f''(x).(v,u)$
\end{itemize}
\end{definition}

Thus the history of linear logic repeated itself. Ehrhard and Regnier \cite{Ehrhard_2003} designed a logic called \emph{differential linear logic} whose models should be differential categories.
We saw in \cref{sec:seely} that the presentation of the semantic might sometimes deviate a bit from the syntax. Something quite similar happens here. Similarly to how the product is lifted through $!$ in \cref{rem:comonoid-from-seely}, the sum (that induces a structure called a \emph{coproduct}) can be lifted through $!$ to produce two natural transformation $\coWeak_X : 1 \arrow !X$ and $\coContr_X : !X \tensor !X \arrow !X$.
\begin{center}
$\coContr_A$: \begin{tikzcd}
!A \tensor !A \arrow[r, "\seelyTwo"] & !(A \with A) \arrow[r, "!(\prodProj_1 + \prodProj_2)"] & !A
\end{tikzcd} \ \ \ \ \ 
$\coWeak_A$: \begin{tikzcd}
1  \arrow[r, "\seelyOne"] & !\top \arrow[r, "!0"] & !X
\end{tikzcd}
\end{center}
Besides, in most differential categories the natural transformation $\dillDerive$ can be derived from a natural transformation $\coDer : X \arrow !X$ as follows.
\begin{center}
$\dillDerive_X:$ \begin{tikzcd}
!X \tensor X \arrow[r, "id_{!X} \tensor \coDer"] & !X \tensor !X \arrow[r, "\coContr"] & !X
\end{tikzcd}
\end{center}
The transformation $\coDer$ can be seen as an operator that computes the derivative in $0$, $f \circ \coDer = f'(0)$. The construction above thus corresponds to the idea that the derivative of $f$ in $x$ can be computed as the derivative in zero of the function $y \mapsto f(x + y)$. Again, there are many flavors of differential categories, see \cite{Blute_2018} for a survey.

Differential logic is a syntactical extension of linear logic that introduces three new rules that account for those morphisms:
a co-weakening (that makes an empty proof), a co-contraction (that sums two proofs) and a co-dereliction (that ``derives'' a proof in 0). 
\begin{center}
(co-weak) \begin{prooftree}
	\hypo{} \infer1{\vdash !A}
\end{prooftree} \ \ \ \ \
(co-contr) \begin{prooftree}
	\hypo{\Gamma_1 \vdash !A} \hypo{\Gamma_2 \vdash !A} \infer2{\Gamma_1, \Gamma_2 \vdash  !A}
\end{prooftree} \ \ \ \ \
(co-der) \begin{prooftree}
	\hypo{\Gamma \vdash A} \infer1{\Gamma \vdash !A}
\end{prooftree}
\end{center}
This syntactical notion of derivation has one primary use called the Taylor expansion. Similarly to how an analytic function (hence a morphism in $\fin(!A,B)$) can be described by its successive derivatives in $0$, the operational behaviour of a proof using the promotion rule can be syntactically described by taking its succesive co-derelictions. It allows to talk about the \emph{ressource sensitivity} of a proof, that is, how a proof can be ``syntactically approximated'' by a proof that uses the formulas less than a given amount of time. This notion nicely relates to \emph{Böhm trees}, a notion in denotational semantic that aims at describing the limit behaviour of a non terminating program. See \cite{Barbarossa_2019} for an enlightening discussion on this topic.

The issue of differential logic though is that the notion of sum is closely tied to non determinism. At first glance, the cut elimination equivalent of the Leibniz rule and the constant rule requires to introduce two new rules.
\begin{center}
(0) \begin{prooftree}
	\hypo{} \infer1{\Gamma \vdash B}
\end{prooftree} \ \ \ \ \
(sum) \begin{prooftree}
	\hypo{\Gamma \vdash B} \hypo{\Gamma \vdash B} \infer2{\Gamma \vdash B}
\end{prooftree}
\end{center}
The first rule introduces a zero proof, the second rule allows to sum two proofs. The first rules makes the logic highly inconsistent as every sequent becomes provable, but it is in fact not that much of an issue because differential logic is more interested in computational properties rather than provability. The second rule however has a lot of repercussions, as it implies that the proof system allows for a non deterministic branching on two possible paths. 
For example, assume that some formula $A$ models the boolean, in the sense that the sequent $\vdash A$ only admits two proofs in linear logic, a proof $true$ and a proof $false$\footnote{In fact $A = 1 \oplus 1$ where $\oplus$ is the disjunction associated to the conjunction $\with$. Since I did not talk about disjunction so I will not give further details}. In differential logic, this sequent admits another proof.
\begin{center}
\begin{prooftree}
	\hypo{true} \infer1{\vdash A} 
	\hypo{false} \infer1{\vdash A}
	\infer2{\vdash A}
\end{prooftree}
\end{center}
whose semantic would be $\semantic{true} + \semantic{false}$, a value that should be neither $\semantic{true}$ neither $\semantic{false}$ but rather a superposition of those two states.

\section{Coherent differentiation} \label{sec:coherent-diff}

In this paper \cite{Ehrhard21}, my supervisor Thomas Ehrhard introduced a new notion called \emph{coherent differentiation}. The cornerstone of coherent differentiation is to generalize differentiation to models that are not always additive categories. The motivation behind this idea is twofold. Firstly, Ehrhard noticed that some models such as \emph{probabilistic coherent spaces}  were not additive categories yet still admitted a notion of differentiation. Secondly, restraining the sum under some compatibility conditions can hopefully make the system deterministic again. For example, there is no morphism such as $\semantic{true} + \semantic{false}$ in probabilistic coherent spaces.

\subsection{Pre-summability structure}

In order to restrict the sum, Erhard introduced a notion called \emph{a pre-summability structure}. Let $\category$ be a model of linear logic such that for any objects $X,Y$, there exists a morphism $0_{X,Y} \in \category(X,Y)$. By abuse of notation, I will usually keep $X$ and $Y$ implicit and write only $0$. Besides, from now on I will start my indices from $0$ in order to be consistent with the notations of the article\footnote{Those notations were subject to debate. My supervisor admitted that he took a bad habit by starting indices by 0. On the other hand, these indices make sense in summability structures because the left term is usually though as a ``zero order term'' while the right term is though as a ``one order term'', as we will see in \cref{sec:coh-diff-operator,sec:monad-S}}.

\begin{definition}
A \emph{pre-summability structure} on $\category$ is a tuple $(\S, \Sproj_0, \Sproj_1, \Ssum)$ such that: \begin{itemize}
\item S is an endofunctor $\S : \category \arrow \category$ that verifies $\S(0_{X,Y}) = 0_{\S X, \S Y}$ for all objects $X,Y$.
\item $\Sproj_0, \Sproj_1, \Ssum : \S \naturalTrans Id$ are natural transformations
\item $\Sproj_0, \Sproj_1$ are jointly monic:  if $f, g \in \homset{Y}{\S X}$, $\disjunctionTwo{\Sproj_0 \circ f = \Sproj_1 \circ g}{\Sproj_1 \circ f = \Sproj_1 \circ g} \imply f = g$ 
\end{itemize}
\end{definition}

 The joint monicity of $\Sproj_0, \Sproj_1$ means that $\S$ is a kind of product: a function $f \in \homset{Y, SX}$ is completely characterized by the two morphisms $\Sproj_0 \circ  f$ and  $\Sproj_1  \circ f$. Note however that contrary to Cartesian product,the pairing of two morphisms $f_0$ and  $f_1$ has no reason to be always defined. When such a pairing exist, we say that $f_0$ and $f_1$ are \emph{summable}
 
\begin{definition}[Sumability] 
$(f_0, f_1)$ are \emph{summable} if there exists a pairing $h$ such that $\begin{cases} \Sproj_0 \circ h = f_0 \\ \Sproj_1 \circ h = f_1 \end{cases}$. As said before when such pairing exists, it is unique. It is written $\Spair{f_0}{f_1}$ and is called the \emph{witness} of the sum. 
The \emph{sum} is then defined as $f_0 + f_1 := \Ssum \Spair{f_0}{f_1}$.
\end{definition}

\begin{remark} \label{rem:pairing-comp-right} Note that by definition of $\Spair{f_0}{f_1}$, $\Sproj_i \circ \Spair{f_0}{f_1} = f_i$ and $\Spair{f_0}{f_1} \circ g = \Spair{f_0 \circ g}{f_1 \circ g}$. Besides, it is easy to check that $\Sproj_0, \Sproj_1$ are summable of witness $id$ and sum $\Ssum$.
\end{remark}

The naturality of $\Sproj_0, \Sproj_1$ and $\Ssum$ express that the sum is compatible with the composition, in a somewhat similar way to additive categories. \begin{itemize}
 	\item The naturality of $\Sproj_0$ and $\Sproj_1$ express that $\S f : \S X \arrow \S Y$ consists in ``applying f coordinate by coordinate'': $\S f \circ \Spair{g_0}{g_1} = \Spair{f \circ g_0}{f \circ g_1}$.
 	\item A consequence of this and \cref{rem:pairing-comp-right} is that if $f_0, f_1$ are sumable, $h \circ f_0 \circ g, h \circ f_1 \circ g$ are summable.  The witness for this sum is $Sh \circ \Spair{f_0}{f_1} \circ g$.
 	\item The naturality of $\Ssum$ then ensures that $h \circ f_0 \circ g+h \circ f_1 \circ g = h \circ (f_0 + f_1) \circ g$. 
\end{itemize}

\begin{result} \label{rem:swith} Let us anticipate a bit over \cref{sec:dill-is-coherent-diff}. This notion of summability generalizes the notion of sum in additive categories. Indeed, let $\category$ be an additive category.
Let $\Swith : \category \arrow \category$ be the functor defined  as $\Swith X := X \with X$ and $\Swith f := f \with f$. This functor can be equiped with the pre-summability structure in which the projections are the projections associated to the cartesian product $\Sproj_i := \prodProj_i : \Swith X \arrow X$ and the sum  is defined as $\Ssum := \prodProj_0 + \prodProj_1: \Swith^2 X \arrow \Swith X$. The projections are jointly monic by unicity of the pairing in a Cartesian product. Finally, $\forall i \in \{0,1\}, \prodProj_i \circ \Swith(0) = 0 = \prodProj_i \circ 0$ so $\S(0) = 0$  by joint monicity of the projections.
 
The sum induced by this presummability structure coincides with the sum of the additive category.  Indeed, any pair of morphisms $f_0, f_1 : X \arrow Y$ are summable (in the sense of the pre-summability structure) with witness $\prodPair{f_0}{f_1}$ and sum $\sigma \circ \prodPair{f_0}{f_1} := (\prodProj_0 + \prodProj_1) \circ \prodPair{f_0}{f_1} = \prodProj_0 \circ \prodPair{f_0}{f_1} +\prodProj_1 \circ \prodPair{f_0}{f_1} = f_0 + f_1$.
\end{result}

\subsection{Summability structure}

The (partial) sum is still quite far from behaving like a sum though, some more axioms are required.  The right axiomatic seems to be the notion of \emph{partial commutative monoid}.

\begin{definition} The sum is called a \emph{partial commutative monoid} if it follows those properties \begin{itemize}
\item Commutativity:  For any morphisms $f, g$ such that $(f, g)$ are summable, $(g, f)$ are summable and $g + f = f + g$
\item Neutrality of 0: For any morphism, $(0, f)$ and $(f, 0)$ are summable and $f + 0 = 0 + f = f$
\item Associativity: For any morphisms $f,g,h$, if $f, g$ and $(f+g),h$ are summable, then $g, h$ are summable, $f, (g+h)$ are summable, and $f+(g+h) = (f+g)+h$
\end{itemize}
A presummability structure whose notion of sum is a partial commutative monoid is called a \emph{summability structure}.
\end{definition}

\begin{result} \label{rem:swith-summability} The presummability structure $\Swith$ defined in \cref{rem:swith} is clearly a summability structure since partial commutative monoids are without a doubt a generalization of commutative monoids.
\end{result}

Interestingly, commutativity implies the existence of a natural transformation $\Ssym = \Spair{\Sproj_1}{\Sproj_0} : \S X \arrow \S X$ and the neutrality of $0$ implies the existence of two natural transformations $\Sinj_0 = \Spair{id}{0} : X \arrow \S X$ and $\Sinj_1 : \Spair{0}{id} : X \arrow \S X$. In fact, Ehrhard shows that those two implications are actually equivalence, but I will not give the proof here. Associativity on the other hand can be reframed as the two following properties.

\begin{definition} We call ($\S$-witness) the following property: for any $X,Y$ and $f, g : X \arrow \S Y$, if $\Ssum \circ f$ and $\Ssum \circ g$ are summable, then $f,g$ are summable.
\end{definition}

\begin{definition} We call  ($\S$-assoc) the following property: assume that $(\Spair{f_0}{f_1}, \Spair{g_0}{g_1})$ are well defined and summable. Then $\forall i \in \{0,1\}, (f_i, g_i)$ are summable, $(f_0 + g_0, f_1 + g_1)$ are summable, and $\Spair{f_0}{f_1} + \Spair{g_0}{g_1} = \Spair{f_0 + g_0}{f_1 + g_1}$. In other words: ``summing two (summable) pairs consists in summing coordinates by coordinates''.
\end{definition}

\begin{proposition} If $0$ is a neutral element, ($\S$-witness) holds and ($\S$-assoc) holds, the sum is associative.
\end{proposition}

I spent some amount of time during my internship to try to re frame associativity in a more structural way, in the vein of neutrality of 0 and commutativity. Doing so might gives some insights on the syntactical equivalent of summability structures. However the property ($\S$-witness) is much subtler than it seems and I start to think that it is by nature quite ``un-structural''. In fact, I think that partial commutative monoids might be a notion that is a bit too restrictive. Indeed, any major notion of coherent differentiation can be defined in a weaker setting, where ($\S$-witness) is replaced by the existence of three operators. I will not give more details as this observation is still an early development and has no strong semantical nor syntactical backing yet. One result that I noticed in this process though is that ($\S$-assoc) always hold.

\begin{result} Any pre-summability structure fulfills ($\S$-assoc).
\end{result}

\begin{proof} By compatibility of addition with regard to composition, $(\Spair{f_0}{f_1}, \Spair{g_0}{g_1})$ summable implies that $(\Sproj_i \circ \Spair{f_0}{f_1}, \Sproj_i \circ \Spair{g_0}{g_1})$ are summable (in other word, $(f_i, g_i)$ are summable). Compatibility also ensures that $ \Sproj_i \circ (\Spair{f_0}{f_1} + \Spair{g_0}{g_1}) = (\Sproj_i \circ\Spair{f_0}{f_1}) + ((\Sproj_i \circ\Spair{f_0}{f_1})) = f_i + g_i$. By definition of summability, it implies that $(f_0 + g_0, f_1 + g_1)$ are summable of witness $\Spair{f_0}{f_1} + \Spair{g_0}{g_1}$. It concludes the proof.
\end{proof}

\subsection{The differentiation operator} \label{sec:coh-diff-operator}

Erhard introduced coherent differentiation as a natural transformation $\devMorphism_X : ! \S  X \arrow \S ! X$. It allows to define for any $f : !X \arrow Y$(e.g a smooth function) a morphism.
\begin{tikzcd}
Df : \ !SX \arrow[r, "\devMorphism_X"] & S !X \arrow[r, "S f"] & SY
\end{tikzcd}. Intuitively, $Df$ is interpreted as the function that takes as input two summable elements $(x, u)$ and return the first order development of $f$ on $x$ for the variation $u$: $Df(x, u) = (f(x), f'(x).u)$. More formally, Erhard asks $\devMorphism$ to follow the rule below
\begin{center}
($\devMorphism$-local)
\begin{tikzcd}
!SX \arrow[r, "\devMorphism_X"] \arrow[rd, "!\Sproj_0"'] & S!X \arrow[d, "\Sproj_0"] \\
                                                         & !X                       
\end{tikzcd}
\end{center}
This rule basically states that $\devMorphism_X = \Spair{!\Sproj_0}{\derive_X} $ for some natural transformation $\derive_X : \ !SX \arrow !X$.So it provides two informations:
\begin{itemize}
 \item There is a \emph{derivation} $\derive_X : \ !SX \arrow !X$.  Intuitively, $f \circ \derive : \ !SX \arrow Y $ takes as input a summable pair $(x, u)$ and returns $f'(x).u$.
\item  $\devFunctor f = Sf \circ \devMorphism_X = \Spair{f \circ !\Sproj_0}{f \circ \derive_X}$, which rewrites intuitively as $Df(x, u) = (f(x), f'(x).u)$. So  the existence of $\devMorphism$ not only ensures that the derivative $f'(x).u$ is defined as soon as $(x, u)$ is summable, but it also ensures that $(f(x), f'(x).u)$ is summable, implying that $f(x) + f'(x).u$ is defined.
\end{itemize}

\begin{result} \label{rem:cohdiff-from-diff}  Recall that differential categories are additive categories, and thus admit a summability structure that coincides with the sum, defined in \cref{rem:swith,rem:swith-summability}. Then the operator $\dillDerive : !X \tensor X \arrow X$ induces a derivation for this summability structure.
\begin{center}
\begin{tikzcd}
\derive_X : \ !\Swith X  \arrow[r, "(\seelyTwo)^{-1}"] & !X \tensor !X  \arrow[r, "id \tensor \der_X"] & !X \tensor X \arrow[r, "\dillDerive_X"] & !X
\end{tikzcd}
\end{center}
It means that a coherent differentiation can be defined as the pairing $\devMorphism_X = \prodPair{!\prodProj_0}{\derive}$. In particular, ($\devMorphism$-local) holds by definition.
\end{result}

\subsection{The monad structure on $\S$} \label{sec:monad-S}

Let us start with a practical reasoning that is fundamental in differential calculus. Take two functions $f$ and $g$ with first order development. Those development can be multiplied.
\begin{align}
f(x+u).g(y+v) & = (f(x) + f'(x).u + o(u)).(g(y)+g'(y).v + o(v)) \\
 &= f(x)g(y) + f(x)g'(y).v + g(y)f'(x).u + (f'(x).u)(g'(y).v) + o(u) + o(v) \\
 &\approx f(x)g(y) + v.f(x)g'(y) + u.f'(x)g(y)
\end{align}
The term $(f'(x).u)(g'(y).v)$ is of second order so it can be neglected. This computation thus shows that the first order variation of $f.g$ around $x$ is $f(x)g'(y).v + g(y)f'(x).u$: this is the so called \emph{Leibniz rule}. Interestingly, the proof above only uses a kind of algebraic reasoning, not how the derivative itself is defined. It means that this reasoning can carry to the more abstract framework of summability structure: the development of $f \tensor g: \ !(X_1 \with X_2) \iso \ !X_1 \tensor !X_2 \arrow Y_1 \tensor Y_2$ can be inferred
from the development of $f$ and the development of $g$, using three categorical constructions:

\begin{itemize}
 \item A natural transformations $\strength^0_{X,Y} := \Spair{id \tensor \Sproj_0}{id \tensor \Sproj_1} : X \tensor SY  \arrow S(X \tensor Y)$ that ``distributes on the right pair'' and a natural transformation $\strength^1_{X,Y} := \Spair{\Sproj_0 \tensor id}{\Sproj_1 \tensor id} : SX \tensor Y  \arrow S(X \tensor Y)$ that ``distributes on the left pair''.  Those two operations are well defined assuming an axiom called ($\S \tensor$-dist) that states the compatibility of the sum with regard to the tensor, in a similar way as in additive categories.
 \item A natural transformation $\Smonadsum_X := \Spair{\Sproj_0 \circ \Sproj_0}{\Sproj_0 \circ \Sproj_1 + \Sproj_1 \circ \Sproj_0} : \S^2 X \arrow \S X$ that keeps the constant factor on the left and sum the two order 1 factors on the right.  I will not give the details but we can show that it is well defined using associativity and compatibility of addition with regard to composition.
\end{itemize}

Those constructions allow to define  $\doubleDistrib: $ \begin{tikzcd}
\S X \tensor \S Y \arrow[r, "\strength^1"] & \S(X \tensor \S Y) \arrow[r, "\S \strength^0"] & \S^2 (X \tensor Y) \arrow[r, "\Smonadsum"] & \S (X \tensor Y)
\end{tikzcd} that ``double distributes the two pairs and stash the order two term'' (note that $\strength^0$ could have been applied before $\strength^1$ for the same outcome). Then the reasoning described above is framed abstractly by the following categorical diagram.

\begin{center}
\begin{tikzcd}
!SX_0 \tensor !SX_1 \arrow[d, "Df \tensor Dg"']  & !(SX_0 \with SX_1) \arrow[l, "(\seelyTwo)^{-1}"'] &  & !S(X_0 \with X_1) \arrow[ll, "!\Spair{S\prodProj_0}{S\prodProj_1}"'] \arrow[d, "\devFunctor (f \tensor g)"] \\
SY_0 \tensor SY_1 \arrow[rrr, "\doubleDistrib"'] &                                                   &  & \S (Y_0 \tensor Y_1)                                                                                       
\end{tikzcd}
\end{center}

We can show that $\Sinj_0 : X \arrow \S X$ and $\Smonadsum_X : \S^2 X \arrow \S X$ gives to $\S$ the structure of a \emph{Monad}. Monads are the dual notion of comonads and  are defined as below.
\begin{definition}
A monad $(\monad, \monadUnit, \monadSum)$ on $\category$ is a functor $\monad : \category \arrow \category$ together with two natural transformations $\monadUnit_A : A \arrow \monad A$ and $\monadSum_A : \monad^2 A \arrow \monad A$ such that the following diagrams commute.
\begin{center}
\begin{tikzcd}
\monad^3 A \arrow[d, "\monadSum_{\monad A}"'] \arrow[r, "\monad \monadSum_A"] & \monad^2 A \arrow[d, "\monadSum_{A}"] \\
\monad^2 A \arrow[r, "\monadSum_A"']                   & !A                      
\end{tikzcd} \ \ \ \ 
\begin{tikzcd}
\monad A \arrow[r, "\monadUnit_{\monad A}"] \arrow[rd, equal] & \monad^2 A \arrow[d, "\monadSum_A" description] & \monad A \arrow[l, "\monad \monadUnit_A"'] \arrow[ld, equal] \\
                                                       & \monad A                                        &                                                      
\end{tikzcd}
\end{center}
\end{definition}

\begin{theorem} \label{thm:S-sm-monad} $(\S, \Sinj_0, \Smonadsum)$ is a monad. 
\end{theorem}

\begin{definition} Similarly to comonads, any monad $\monad$ on $\category$ induces a category $\kleisli$ called the \emph{Kleisli category of the monad} where $\kleisli(X, Y) = \category(X, \monad Y)$ (same as for comonads but $\monad$ is on the right). There is also a faithfull functor $\kleisliCast : \category \arrow \kleisli$. I will not give the exact constructions as they are quite similar to the constructions of the comonads.
\end{definition}

Those categories are extensively studied in semantic. The reason is that they model \emph{effects}. For example, the \texttt{Option} type in programming language (an element of \texttt{Option X} is either $\texttt{Some}(x)$ where $x \in X$ or \texttt{None}) is a monad. The Kleisli category of this monad models computations that can either terminate correctly (by outputting $\texttt{Some}(result)$) or raise an exception (by outputting \texttt{None}). There are also monads that models memory, non determinism, randomness, etc... It suggests that differentiation might be an effect too. As a result, most of the last part of the internship was centered around the study of $\kleisliS$, see \cref{sec:distributive-law}.

\subsection{The axioms of coherent differentiation} \label{sec:cohdiff-axioms}

  Those diagrams turned out to corresponds to the elementary properties that derivation should follows, and their interpretation coincides with the interpretion of the diagram required in DiLL.  I will not give the diagram themselves, rather their interpretation.
\begin{itemize}
 \item[($\devMorphism$-lin)]Two axioms that give the interactions between $\devMorphism$ and $\Sinj_0$/$\Smonadsum$. They state that the derivative of a function is linear. The first one ($\devMorphism$-lin-1) ensures  that $f'(x).0 = 0$ and the second one ($\devMorphism$-lin-2) ensures that $f'(x).(u_1 + u_2) = f'(x).u_1 + f'(x).u_2$. 
 \item[($\devMorphism$-chain)]Two axioms that give the interactions between $\devMorphism$ and $\der$/$\dig$ The axiom ($\devMorphism$-chain-1) states that the derivative of a linear function is the function itself. The axiom ($\devMorphism$-chain-2) states the chain rule.
 \item[($\devMorphism$-Leibniz)]Two axioms that give the interaction between $\devMorphism$ and $\weak$/$\contr$. The axiom ($\devMorphism$-Leibniz-1) states that the derivative of a constant function is null. The axiom ($\devMorphism$-Leibniz-2) is the Leibniz rule, and is equivalent to diagram motivated in \cref{sec:monad-S}.
 \item[($\devMorphism$-Schwarz)]One axiom that gives the interaction of $\devMorphism$ with itself.  This axiom states the Schwarz rule.
\end{itemize}

Just as an example, here is the axiom ($\devMorphism$-lin).
\begin{center}
($\devMorphism$-lin)
\begin{tikzcd}
!SX \arrow[r, "\devMorphism_X"] \arrow[rd, "\der_{SX}"'] & S!X \arrow[d, "S\der_X"] \\
                                                        & SX                      
\end{tikzcd}
\begin{tikzcd}[column sep = large]
!\S X \arrow[r, "\devMorphism_X"] \arrow[rd, "\der_{\S X}"'] & S!X \arrow[d, "\S \der_X"] \arrow[rd, "S (f \circ \der)", dashed] &      \\
                                                             & \S X \arrow[r, "\S f"', dashed]                                   & \S Y
\end{tikzcd}
\end{center}
How does this rule says anything about the derivative of a linear function ? Let us plug a morphism of the shape $Sf$ at the end of the diagram. Then by functoriality of $\S$, $\S f \circ \S \der = \S (f \circ \der)$.

Recall that if $f : X \arrow Y$, then $f \circ \der : !X \arrow Y$ can be interpreted as the linear function that maps an element $x \in X$ to $f(x) \in Y$. Thus, thanks to our interpretation of $\devMorphism$, the the top path of the diagram $\S (f \circ \der) \circ \devMorphism_X$ consists in a morphism that maps a pair $(x, u)$ to the first order development $(f(x), f'(x).u)$. On the other hand, the bottom path $\S f \circ \der_{\S X}$ consists in the linear function that maps a pair $(x, u)$ to the pair $(f(x), f(u))$. So the diagram above interprets as: ``if f is a linear function, for any $x, u \in X$ that are summable, $f'(x).u = f.u$''.

\section{Differentiation in differential linear logic is a particular case of coherent differentiation} \label{sec:dill-is-coherent-diff}

The first objective of the internship was to show that the structure of a differential categories induces a structure of coherent differentiation. The goal behind this was threefold. Firstly, it is a nice starting point for an internship, as it allowed me to hone my understanding of both models of differential linear logic and coherent differentiation. Secondly, it ensures that coherent differentiation is a generalization of differential categories, not a divergent notion. Finally, it implies that understanding coherent differentiation might bring some insights to differential categories too.

I already showed that the sum in additive categories induce a summability structure (see \cref{rem:swith,rem:swith-summability}) and that any differential operator $\dillDerive : !X \tensor X \arrow X$ induces a candidate for coherent differentiation (see \cref{rem:cohdiff-from-diff}). What remains is to check the various axioms of coherent differentiation. Unsurprisingly, each axiom on $\devMorphism$ ends up being a consequence of its counterpart on $\dillDerive$. For example, the axiom ($\devMorphism$-chain-1) that states that ``the derivative of a linear function is the function itself'' is a consequence of the interaction between $\dillDerive$ and $\der$ that states the same idea. The only exceptions are the two the axioms ($\devMorphism$-lin) that have a more structural counterpart expressed in the fact that $\derive$ is defined from $\dillDerive$ with a dereliction (so, by forgetting the linearity in the second coordinate).

The proof is conceptually not that hard, but many technicalities arose. Consequently, the full redacted proof is 25 pages long. Obviously, I will not write everything down here. I will rather exhibit the proof structure through the example of ($\devMorphism$-chain-1), as this rule is probably the easiest one that still contains some structural arguments. The first crucial step was to show that this rule was equivalent to a rule about $\derive$.

\begin{result} \label{res:axiom-derive} For any summability structure and morphism $\devMorphism$ that follows ($\devMorphism$-local), the diagram ($\devMorphism$-lin) is equivalent to
\begin{center}
\begin{tikzcd}
!SX \arrow[r, "\derive_X"] \arrow[d, "!\Sproj_1"'] & !X \arrow[d, "\der_X"] \\
!X \arrow[r, "\der_X"']                            & X                   
\end{tikzcd}
\end{center}
\end{result}

\begin{proof} By joint monicity of $\Sproj_0, \Sproj_1$, the diagram ($\devMorphism$-lin) holds if and only if the diagram below on the left holds for any $i \in \{0,1\}$. This diagram admits the diagram chase on the right.
\begin{center}
\begin{tikzcd}
!\S X \arrow[d, "\der_{\S X}"'] \arrow[rr, "\devMorphism_X"] &   & \S ! X \arrow[d, "\S \der_X"] \\
\S X \arrow[r, "\Sproj_i"']                                  & X & \S X \arrow[l, "\Sproj_i"]   
\end{tikzcd}
\begin{tikzcd}
!\S X \arrow[dd, "\der_{\S X}"'] \arrow[rrr, "\devMorphism_X"] \arrow[rd, "!\Sproj_i"'] \arrow[rdd, "(a)" description, phantom] &                         &                        & \S ! X \arrow[dd, "\S \der_X"] \arrow[ld, "\Sproj_i"] \arrow[ldd, "(b)" description, phantom] \\
                                                                                                                               & !X \arrow[d, "\der_X"'] & !X \arrow[d, "\der_X"] &                                                                                              \\
\S X \arrow[r, "\Sproj_i"']                                                                                                    & X \arrow[equal]{r}             & X                      & \S X \arrow[l, "\Sproj_i"]                                                                  
\end{tikzcd}
\end{center}
The commutation (a) is the naturality of $\der$ and the commutation (b) is the naturality of $\Sproj_i$.
This reduced diagram trivially holds for $i=0$ using that $\Sproj_0 \circ \devMorphism = !\Sproj_0$. The induced diagram for $i=1$ on the other hand is exactly the diagram of the result, using that $\Sproj_1 \circ \devMorphism = \derive$.
\end{proof}

This first step is crucial, because this diagram is already much closer to its counterpart on $\dillDerive$. Here is a side by side comparison.
\begin{center}
\begin{tikzcd}
! X \with X \arrow[r, "\derive_X"] \arrow[d, "!\prodProj_1"'] & !X \arrow[d, "\der_X"] \\
!X \arrow[r, "\der_X"']                            & X                   
\end{tikzcd} \ \ \ \ \ 
\begin{tikzcd}
!X \tensor X \arrow[r, "\dillDerive_X"] \arrow[d, "\weak_X \tensor id_X"'] & !X \arrow[d, "\der_X"] \\
1 \tensor X \arrow[r, "\tensorUnitL"']                                     & X                     
\end{tikzcd}
\end{center}
The biggest difference is that the ``structural work'' is done on $\with$ on the left and on $\tensor$ on the right. Fortunately, recall that the Seely isomorphisms relate those two structures.

\begin{result}[Rewriting the projections] \label{res:proj-struct} The following diagrams commute
\begin{center}
\begin{tikzcd}
!(X \with Y) \arrow[d, "(\seelyTwo)^{-1}"'] \arrow[r, "! \prodProj_0"] & !X                                      \\
!X \with !Y \arrow[r, "id_{!X} \tensor \weak_Y"']                  & !X \tensor 1 \arrow[u, "\tensorUnitR"']
\end{tikzcd}
\begin{tikzcd}
!(X \ Y) \arrow[d, "(\seelyTwo)^{-1}"'] \arrow[r, "! \prodProj_1"] & !Y                                      \\
!X \tensor !Y \arrow[r, "\weak_X \tensor id_{!Y}"']                  & 1 \tensor !Y \arrow[u, "\tensorUnitL"']
\end{tikzcd}
\end{center}
\end{result}

\begin{proof} I do the proof for the second diagram. Unfolding the definition of the weakening (recall \cref{rem:comonoid-from-seely}) leads up to the following diagram chase.
\begin{center}
\begin{tikzcd}
!(X \with Y) \arrow[rr, "\prodProj_1"] \arrow[dd, "(\seelyTwo)^{-1}"'] \arrow[rd, "!(\comonoidUnit_X \with id_{!Y})"] \arrow[rdd, "(a)", phantom] & {}                                                                                                   & !Y \arrow[ldd, "(b)", phantom]           \\
                                                                                                                                                   & !(\top \with Y) \arrow[d, "(\seelyTwo)^{-1}"'] \arrow[ru, "!\prodProj_1"] &                                          \\
!X \tensor !Y \arrow[r, "!\comonoidUnit_X \tensor id_{!Y}"']                                                                                       & !\top \tensor !Y \arrow[r, "\seelyOne^{-1}"']                                                        & 1 \tensor !Y \arrow[uu, "\tensorUnitL"']
\end{tikzcd}
\end{center}
The commutation (a) is the naturality of $\seelyTwo$, the commutation (b) is one of the diagrams that ensures that $(!, \seelyOne, \seelyTwo)$ is a symetric monoidal functor. The triangle at the top can be directly computed.
\end{proof}

We now have all of the ingredients necessary to prove that ($\devMorphism$-chain-1) holds.

\begin{result} \label{res:chain-1} The diagram of \cref{res:axiom-derive} commutes, so ($\devMorphism$-chain-1) holds.
\end{result}

\begin{proof}  We perform the following diagram chase.
\begin{center}
\begin{tikzcd}
! (X \with X) \arrow[r, "(\seelyTwo)^{-1}"] \arrow[d, "!\prodProj_1"'] \arrow[rd, "(a)", phantom] & !X \tensor !X \arrow[r, "!X \tensor \der_X"] \arrow[d, "\weak_X \tensor id_{!X}" description] \arrow[rd, "(b)", phantom] & !X \tensor X \arrow[r, "\dillDerive_X"] \arrow[d, "\weak_X \tensor id_X" description] \arrow[rd, "(c)", phantom] & !X \arrow[d, "\der_X"] \\
!X \arrow[r, "\tensorUnitL^{-1}"']                                                                 & 1 \tensor !X \arrow[r, "id_1 \tensor \der_X"']                                                                           & 1 \tensor X \arrow[r, "\tensorUnitL"']                                                                           & X                     
\end{tikzcd}
\end{center}
The commutation (a) comes from \cref{res:proj-struct}, the commutation (b) is the functoriality of $\tensor$, and the commutation (c) is the linear rule in models of differential linear logic.
\end{proof}

To sum up, I proved that the different axioms of coherent differentiation hold by using the following proof method. \begin{itemize}
 \item I proved that all of the axioms of coherent differentiation are equivalent to axioms on the derivation operator $\derive$ that is introduced by $(\devMorphism$-local), as in \cref{res:axiom-derive}
 \item Those axioms are always similar to their counterpart in differential categories, except that the structural work is done on $\with$ rather than on $\tensor$. Fortunately I showed that the properties of the Seely isomorphisms stated in \cref{def:seely-isomorphisms} close the bridge between those two worlds, as in \cref{res:proj-struct}.
 \item I wrapped everything up using basic properties such as naturality and functoriality, as in \cref{res:chain-1}.
\end{itemize}

I cannot resist to give one last structural commutation as this one is particularly insightful on the fact that the Seely isomorphisms lift the structure of the Cartesian product $\with$ under $!$ to give to the tensor $\tensor$ a structure similar to a Cartesian product on the objects of shape $!X$, accordingly to \cref{thm:sm-is-cartesian} and \cref{rem:comonoid-from-seely}.
\begin{center}
\begin{tikzcd}
!X \arrow[d, "!\prodPair{f_0}{f_1}"'] \arrow[r, "\contr_X"] & !X \tensor !X \arrow[d, "!f_0 \tensor !f_1"] \\
!(Y_0 \with Y_1) \arrow[r, "(\seelyTwo)^{-1}"']               & !Y_1 \tensor !Y_1                           
\end{tikzcd}
\end{center}

\section{Distributive laws, distributive laws everywhere} \label{sec:distributive-law}

 In his paper, Ehrhard noticed that all of the concrete instances of summability structure he found in the various models consisted in a functor $\canonicalS = I \linarrow \_$ where $I = 1  \with 1$. The monads of this shape are said to be ``right adjoints'' to a comonad of shape $\canonicalCoS = \_ \tensor I$, a relation written $\canonicalCoS \dashv \canonicalS$ that I will not detail here. Ehrhard then showed that this relation relates coherent differentiation on $\canonicalS$ to what is called a coalgebra on $I$, that is a morphism $\delta : I \arrow !I$ with some properties. It is somewhat folklore in the field that under those conditions, the Kleisli category of $\canonicalCoS$ are still models of linear logic, as  it was already answered by Girard in a particular example. The initial goal of this last part of the internship was to make a definitive proof of this statement.
 
 In parallel, Ehrhard noticed that the rules ($\devMorphism$-lin) and ($\devMorphism$-chain) made $\devMorphism$ what is called a \emph{distributive law} from the monad $\S$ to the comonad $!$, a notion introduced a while ago by Beck \cite{Beck_1969} and quite studied in some modern fields of semantic, see \cite{Power_2002} for example. I then noticed that a similar notion of distributive law existed, this time from a comonad to another comonad. This notion was surprisingly very close to the commutations required on $\delta$. This is when I noticed that Ehrhard implicitely showed in his paper that the adjunction $\canonicalCoS \dashv \canonicalS$ carries distributive laws $\devMorphism$ from the monad $\canonicalCoS$ to the comonad $!$ to distributive laws from the comonad $\canonicalCoS$ to the comonad $!$, and that the existence of such distributive law is equivalent to the existence of a coalgebra $\delta$.
 
The fact that the kleisli category $\kleisliConstructor{\category}{\canonicalCoS}$ is a model of linear logic thus seemed to hint that $\kleisliS$ should be a model of a linear logic as well (and that a proof of the first statement should carry to the second through an adjunction). So I decided to devise some generic conditions under which the Kleisli category of a monad (or a comonad) is a model of linear logic to make everything clearer.

Recall that a model of linear logic $\category$ is characterized by nothing more than the existence of some functors together with the existence of some natural transformations that follow some commutations: a functor $\tensor$ together with four natural transformations $\tensorAssoc, \tensorUnitL, \tensorUnitR, \tensorSym$, a functor $!$ together with two natural transformations $\der, \dig$, and two natural transformations $\seelyOne : \top \arrow !1$ and  $\seelyTwo_{X,Y} : !X \tensor !Y \arrow !(X \with Y)$. Finally, there are two universal mapping properties (the Cartesian product $\with$ and the closure with regard to $\tensor$). I will not talk about those here, but they can also be looked upon this point of view using the fact that giving an universal mapping property is the same as giving two functors and two natural transformations (thanks to the notion of \emph{adjunction}).

Now let us assume that we can extend those functors to $\kleisliS$ in the sense of \cref{def:extension} below and that we can extend those natural transformations in the sense of \cref{def:extension-natural}.

\begin{definition} \label{def:extension} Given a monad $\monad[1]$ on a category $\category[1]$ and a monad $\monad[2]$ on a category $\category[2]$, a functor $\kleisliExtension{F} : \kleisli[1] \arrow \kleisli[2]$ \emph{lifts} the functor $F : \category[1] \arrow \category[2]$ to the Kleisli categories if for any object $X$ of $\category[1]$, $\kleisliExtension{F} X = F X$, and if for any morphism $f \in \category(X, Y)$, $\kleisliExtension{F}(\kleisliCast[1](f)) = \kleisliCast[2](Ff)$. In other word, ``$\kleisliExtension{F}$ coincides with $F$ on $\category[1]$''.
\end{definition}

\begin{definition} \label{def:extension-natural} Take two functors $F, G : \category[1] \arrow \category[2]$ with some extensions $\kleisliExtension{F}, \kleisliExtension{G} : \kleisli[1] \arrow \kleisli[2]$. Given a natural transformation $\alpha : F \naturalTrans G$, we can define a family of morphisms
\[ \kleisliCast[2](\alpha) := (\kleisliCast[2](\alpha_X))_X \in  \kleisli[2](\kleisliExtension{F}X, \kleisliExtension{G}X) \]
We say that $\alpha$ extends to $\kleisliExtension{F}$ and $\kleisliExtension{G}$ if $\kleisliCast[2](\alpha)$ is a natural transformation $\kleisliExtension{F} \naturalTrans \kleisliExtension{G}$.
\end{definition}

Then all the commutations on the natural transformations $\kleisliCastS(.)$ would be obtained for free using the functoriality of $\kleisliCastS$. So $\kleisliS$ would ``inherits'' from $\category$ the necessary structure to be a model of linear logic. The litterature is already on point on the necessary and sufficient conditions that allows to extend a functor, see \cite{Mulry_1994} for example.

\begin{definition} A \emph{distributive law} from a functor F to two monads $(\monad[1], \monadUnit[1], \monadSum[1])$ and $(\monad[2], \monadUnit[2], \monadSum[2])$ is a natural transformation $\kleisliLift : F \monad[1] \naturalTrans \monad[2] F$ such that the two diagrams commute.
\begin{center}
(Lift-Unit) 
\begin{tikzcd}
FX \arrow[rd, "{\monadUnit[2]}"'] \arrow[r, "{F \monadUnit[1]}"] & F \monad[1] X \arrow[d, "\kleisliLift"] \\
                                                           & \monad[2] F X                          
\end{tikzcd} \ \ 
(Lift-Sum)
\begin{tikzcd}
F \monad[1]^2 X \arrow[d, "{\monadSum[1]}"'] \arrow[r, "{\kleisliLift_{\monad[1] X}}"] & \monad[2] F \monad[1] X \arrow[r, "{\monad[2] \kleisliLift_X}"] & \monad[2]^2 F X \arrow[d, "{\monadSum[2]}"] \\
F \monad[1] X \arrow[rr, "\kleisliLift"']                                          &                                                               & \monad[2] F X                            
\end{tikzcd}
\end{center}
We write $(F, \kleisliLift) : (\category[1], \monad[1]) \arrow (\category[2], \monad[2])$ if $\kleisliLift$ is a distributive law of F on  $\monad[1]$ and $\monad[2]$
\end{definition}

This notion is also called distributive law, because a distributive law from a monad $\monad$ to a comonad $!$ is in fact a distributive law from the functor $F=!$ to the monads $\monad[1] = \monad[2] = \monad[0]$ that fulfills some additional commutations. This intersection of terminology was somewhat troublesome because it made my bibliographic work much harder.

\begin{theorem} \label{thm:extension-and-kleisli-lifting} There is a bijection between the set of functors $\kleisliExtension{F} : \kleisli[1] \arrow \kleisli[2] $ that extends F and the set of distributive laws $\kleisliLift : F \monad[1] \naturalTrans \monad[2] F$
\end{theorem}

However, I did not see anything in the litterature (at first) that characterizes when a natural transformation extends to the Kleisli category. Surprisingly, I discovered during my internship this very straightforward characterization.

\begin{definition} Take $(F, \kleisliLift_F), (G, \kleisliLift_G): (\category[1], \monad[1]) \arrow (\category[2], \monad[2])$ two distributive laws. A morphism of distributive laws $\alpha : (F, \kleisliLift_F) \arrow (G, \kleisliLift_G)$ is a natural transformations $\alpha : F \naturalTrans G$ that verifies the following diagram
\begin{center}
\begin{tikzcd}
F \monad[1] \arrow[d, "\kleisliLift^F"'] \arrow[r, "{\alpha \monad[1]}"] & G \monad[1] \arrow[d, "\kleisliLift^G"] \\
\monad[2] F \arrow[r, "{\monad[2] \alpha}"']                             & \monad[2] G                            
\end{tikzcd}
\end{center}
\end{definition}

\begin{result} \label{thm:extension-and-kleisli-lifting-morphism} Take $(F, \kleisliLift_F), (G, \kleisliLift_G): (\category[1], \monad[1]) \arrow (\category[2], \monad[2])$ two distributive laws with associated extension $\kleisliExtension{F}, \kleisliExtension{G}$. Take $\alpha : F \naturalTrans G$ a natural transformation. The the following are equivalent: \begin{itemize}
\item[(1)] $\alpha : (F, \kleisliLift_F) \arrow (G, \kleisliLift_G)$ is a morphism of distributive laws
\item[(2)] $\alpha$ extends to the kleisli category, in other words $\kleisliCast[2](\alpha) : \kleisliExtension{F} \naturalTrans \kleisliExtension{G}$ is a natural transformation
\end{itemize}
\end{result}

The proof is quite simple (just unfold the definitions and use naturality), yet this result provides a powerful toolbox to extend natural transformations. So it \emph{had} to be an already existing notion. It turns out to be the case: the notion of morphism of distributive law is introduced in a more abstract setting in~\cite{Street_1972}. Besides, an article by Power and Watabane \cite{Power_2002} brought this notion to the semantic community. A recent discussion with Paul André Meliès suggested that this fact is somewhat folklore in some communities, but also that this notion seems to be somewhat underutilized in concrete settings such as mine. So the conclusions I draw below seem to be new and are quite enlightening.

\begin{itemize}
\item Applying the method to the ressource comonad gives back the notion of distributive law between a monad and a comonad. For example, in the setting of coherent differentiation, ($\devMorphism$-lin) are the rules required for $\devMorphism$ to be a distributive law between the monad $\S$ and the functor $!$, and ($\devMorphism$-chain) are the rules required to extend $\der$ and $\dig$. 

\item  The functor $\tensor$ extends to $\kleisli$ if and only if there exists a distributive law $\kleisliLift^{\tensor}_{X,Y} : \monad X \tensor \monad X \arrow \monad(X \tensor Y)$. Interestingly, the commutations required to extend the natural transformations $\tensorAssoc, \tensorUnitL, \tensorUnitR, \tensorSym$ are exactly the commutations that states that $(\monad, \monadUnit, \kleisliLift^{\tensor})$ is a lax symetric monoidal functor from $(\category, \tensor, 1)$ to itself. Adding the commutations required for $\kleisliLift^{\tensor}$ to be  a distributive laws exactly gives the conditions required to makes $\monad$ what is called a \emph{symetric monoidal monad}. This notion has been developped independently precisely to show when the Kleisli category of a monad inherits the structure of a symetric monoidal category, see \cite{Kock_1970}. Besides, $\S$ precisely fulfills those conditions for $\kleisliLift^{\tensor} = \doubleDistrib$ (defined in \cref{sec:monad-S}). Thus $\kleisliS$ inherits from $\category$ the structure of a symetric monoidal category. 


\item Finally, it turns out that in the setting of coherent differentiation, the conditions to extend $\seelyOne$ and $\seelyTwo$ are exactly the two ($\devMorphism$-Leibniz) axioms.

\end{itemize}

To conclude, I built some generic tools that exhibit that the axioms of coherent differentiation (excluding the Schwarz rule) turn out to be the necessary and sufficient conditions to ensures that $\kleisliS$ inherits from $\category$ the structure of a model of linear logic. So the abstract principles of first order differential calculus turns out to be \emph{closely} tied to the very generic notion of structure extension applied to the specific monad $(\S, \Sinj_0, \Smonadsum)$.

There is dually a notion of distributive law from a functor to two \emph{comonads}, as well as similar notion of morphism between such distributive laws. Applying this notion to the comonad $\_ \tensor I$ should solve the initial goal of showing that the Kleisli category of this comonad is a model of linear logic, but I was a bit short on time so I could not write it down yet. Besides, I think that the adjunction $\canonicalCoS \dashv \canonicalS$ should canonically relate distributive laws to monads and their morphisms with distributive laws to comonads and their morphisms. Most of this work is already implicitely done in Ehrhard's paper.

\section{Conclusion}

To sum up most of my internship, I showed in \cref{sec:dill-is-coherent-diff} that coherent differentiation is a generalization of differential categories and that its axioms except ($\devMorphism$-lin) are in a one to one correspondence to the axioms on differential categories.
I built in \cref{sec:distributive-law} some generic tools that, when applied to coherent differentiation, exhibit that those axioms (excluding the Schwarz rule) turn out to be the necessary and sufficient conditions to ensures that $\kleisliS$ inherits from $\category$ the structure of a model of linear logic. 

\paragraph{Plan for future work}

Recall that all of the concrete instances of summability structure Ehrhard found in the various models actually consists in the functor $\canonicalS = 1 \with 1 \linarrow \_$. A natural question is whether or not a summability structure is always of this shape. A step in the right direction would be to find instances of Monads $\monad$ that are not of shape $1 \with 1 \linarrow \_$ such that $\monad$ follows the properties discussed in \cref{sec:distributive-law}, even if they does not correspond to a notion of differentiation. But those considerations require to learn a bunch of different models, this is why the question was not studied during the internship.

Finally, we expect that coherent differentiation might relate to the growing field of automated differentiation, used in machine learning. The difference is that automated differentiation only differentiates with regard to a type of real numbers, while coherent differentiation differentiates with regard to everything (which is admittedly a weird thing). Further work has to be done in order to relate those two fields, as it could provide an exciting application 

\printbibliography

@inproceedings{Mellies_2009,
  title={CATEGORICAL SEMANTICS OF LINEAR LOGIC},
  author={Paul-Andr{\'e} Melli{\`e}s},
  year={2009}
}

@ARTICLE{Lane_1963,
    author = {Saunders Mac Lane},
    title = {Natural associativity and commutativity},
    journal = {Rice University Studies},
    year = {1963}
}

@article{Kelly_1964,
title = {On MacLane's conditions for coherence of natural associativities, commutativities, etc.},
journal = {Journal of Algebra},
volume = {1},
number = {4},
pages = {397-402},
year = {1964},
issn = {0021-8693},
doi = {https://doi.org/10.1016/0021-8693(64)90018-3},
url = {https://www.sciencedirect.com/science/article/pii/0021869364900183},
author = {G.M Kelly}
}

@article{Girard_1986,
title = {The system F of variable types, fifteen years later},
journal = {Theoretical Computer Science},
volume = {45},
pages = {159-192},
year = {1986},
issn = {0304-3975},
doi = {https://doi.org/10.1016/0304-3975(86)90044-7},
url = {https://www.sciencedirect.com/science/article/pii/0304397586900447},
author = {Jean-Yves Girard}
}

@article{Girard_1987,
title = {Linear logic},
journal = {Theoretical Computer Science},
volume = {50},
number = {1},
pages = {1-101},
year = {1987},
issn = {0304-3975},
doi = {https://doi.org/10.1016/0304-3975(87)90045-4},
url = {https://www.sciencedirect.com/science/article/pii/0304397587900454},
author = {Jean-Yves Girard}
}

@inproceedings{Seely_1989,
    address = {Boulder, Colorado},
    series = {Contemporary {Mathematics}},
    title = {Linear {Logic}, $*$-{Autonomous} {Categories} and {Cofree} {Coalgebras}},
    volume = {92},
    booktitle = {Categories in {Computer} {Science} and {Logic}},
    publisher = {American Mathematical Society},
    author = {Seely, R. A. G.},
    year = {1989},
    pages = {371--382}
}

@article{Ehrhard_2003,
  TITLE = {{The differential lambda-calculus}},
  AUTHOR = {Ehrhard, Thomas and Regnier, Laurent},
  URL = {https://hal.archives-ouvertes.fr/hal-00150572},
  NOTE = {41 pages},
  JOURNAL = {{Theoretical Computer Science}},
  PUBLISHER = {{Elsevier}},
  VOLUME = {309},
  NUMBER = {1-3},
  PAGES = {1-41},
  YEAR = {2003},
  MONTH = Dec,
  DOI = {10.1016/S0304-3975(03)00392-X},
  HAL_ID = {hal-00150572},
  HAL_VERSION = {v1},
}

@misc{Ehrhard21,
  doi = {10.48550/ARXIV.2107.05261}, 
  url = {https://arxiv.org/abs/2107.05261},
  author = {Ehrhard, Thomas},
  title = {Coherent differentiation},
  publisher = {arXiv},
  year = {2021},
  copyright = {arXiv.org perpetual, non-exclusive license}
}

@misc{Blute_2018,
  doi = {10.48550/ARXIV.1806.04804},  
  url = {https://arxiv.org/abs/1806.04804},  
  author = {Blute, R. F. and Cockett, J. R. B. and Lemay, J-S. Pacaud and Seely, R. A. G.},  
  title = {Differential Categories Revisited},  
  publisher = {arXiv},  
  year = {2018}, 
  copyright = {arXiv.org perpetual, non-exclusive license}
}

@InProceedings{Mulry_1994,
	author="Mulry, Philip S.",
	editor="Brookes, Stephen and Main, Michael and Melton, Austin and Mislove, Michael and Schmidt, David",
	title="Lifting theorems for Kleisli categories",
	booktitle="Mathematical Foundations of Programming Semantics",
	year="1994",
	publisher="Springer Berlin Heidelberg",
	address="Berlin, Heidelberg",
	pages="304--319",
	isbn="978-3-540-48419-6"
}

@article{Street_1972,
title = {The formal theory of monads},
journal = {Journal of Pure and Applied Algebra},
volume = {2},
number = {2},
pages = {149-168},
year = {1972},
issn = {0022-4049},
doi = {https://doi.org/10.1016/0022-4049(72)90019-9},
url = {https://www.sciencedirect.com/science/article/pii/0022404972900199},
author = {Ross Street}
}

@article{Power_2002,
title = {Combining a monad and a comonad},
journal = {Theoretical Computer Science},
volume = {280},
number = {1},
pages = {137-162},
year = {2002},
note = {Coalgebraic Methods in Computer Science},
issn = {0304-3975},
doi = {https://doi.org/10.1016/S0304-3975(01)00024-X},
url = {https://www.sciencedirect.com/science/article/pii/S030439750100024X},
author = {John Power and Hiroshi Watanabe}
}

@article{Ehrhard_2005,
  TITLE = {{Finiteness spaces}},
  AUTHOR = {Ehrhard, Thomas},
  URL = {https://hal.archives-ouvertes.fr/hal-00150276},
  NOTE = {32 pages},
  JOURNAL = {{Mathematical Structures in Computer Science}},
  PUBLISHER = {{Cambridge University Press (CUP)}},
  VOLUME = {15},
  NUMBER = {4},
  PAGES = {615-646},
  YEAR = {2005},
  MONTH = Jul,
  DOI = {10.1017/S0960129504004645},
  HAL_ID = {hal-00150276},
  HAL_VERSION = {v1},
}

@InProceedings{Beck_1969,
author="Beck, Jon",
editor="Eckmann, B.",
title="Distributive laws",
booktitle="Seminar on Triples and Categorical Homology Theory",
year="1969",
publisher="Springer Berlin Heidelberg",
address="Berlin, Heidelberg",
pages="119--140",
isbn="978-3-540-36091-9"
}

@article{Barbarossa_2019,
author = {Barbarossa, Davide and Manzonetto, Giulio},
title = {Taylor Subsumes Scott, Berry, Kahn and Plotkin},
year = {2019},
issue_date = {January 2020},
publisher = {Association for Computing Machinery},
address = {New York, NY, USA},
volume = {4},
number = {POPL},
url = {https://doi.org/10.1145/3371069},
doi = {10.1145/3371069},
journal = {Proc. ACM Program. Lang.},
month = {dec},
articleno = {1},
numpages = {23}
}

@article{Kock_1970,
  title={Monads on symmetric monoidal closed categories},
  author={Anders Kock},
  journal={Archiv der Mathematik},
  year={1970},
  volume={21},
  pages={1-10}
}

\newpage
\section*{Appendix: basic category theory definition} \label{sec:annex}

I assumed in this report that the reader knows the very basic definitions of a category (the . This appendix is here just in case to give a refresher on those notions.

\begin{definition}[Category] A category $\category$ consists in \begin{itemize}
	\item In a set $Obj(\category)$ called the set of objects of $\category$
	\item For any pair of objects $A,B \in Obj(\category)$, a set $\category(A, B)$ called the set of morphisms from $A$ to $B$.
\end{itemize}
Such that : \begin{itemize}
	\item For any object $A$, there exists a morphism $id^{\category}_A \in \category(A, A)$
	\item For any morphisms $f \in \category(A, B)$ and $g \in \category(B, C)$, there exists a morphism $g \circ_{\category} f \in \category(A, C)$
	\item For any $f \in \category(A, B)$, $f \circ_{\category} id^{\category}_A = id^{\category}_B \circ_{\category} f = f$
	\item For any  $f \in \category(A, B)$, $g \in \category(B, C)$ and $h \in \category(C, D)$, $h \circ_{\category} (g \circ_{\category} f) = (h \circ_{\category} g) \circ_{\category} f$
\end{itemize}
\end{definition}

\begin{remark} I directly introduce $Obj(\category)$ and $\category(X, Y)$ as sets rather than classes because I do not want to deal with foundational issues here (the set of all set is not defined, the category of all category is not defined, etc).
\end{remark}

\begin{notation} When there is no doubt about the category $\category$ considered, I will write $\circ$ for $\circ_{\category}$ and $id$ for $id^{\category}$. Besides, I will often write $f : A \arrow B$ for $f \in \category(A, B)$ (taking inspiration from functional notations).
\end{notation}

\begin{definition}[Isomorphism] A morphism $f : A \arrow B$ is an \emph{isomorphism} if there exists a morphism $g : B \arrow A$ such that $f \circ g = id_B$ and $g \circ f = id_A$.
\end{definition}

\begin{definition}[Functor] Given two categories $\category[1]$ and $\category[2]$, a functor $F : \category[1] \arrow \category[2]$ consists in: \begin{itemize}
	\item A function $F^{Obj} : Obj(\category[1]) \arrow Obj(\category[2])$
	\item For any objects $A, B \in Obj(\category[1])$, a function $F_{A,B} : \category[1](A, B) \arrow \category[2](F^{Obj}(A), F^{Obj}(B))$ such that $F_{A,A}(id^{\category[1]}_A) = id^{\category[2]}_{F^{Obj}(A)}$ and $F_{A,C} (f \circ_{\category[1]} g) = F_{A, B}(f) \circ_{\category[2]} F_{B,C}(g)$.
\end{itemize}
We will often write $F$ for both $F^{Obj}$ and $F_{A,B}$ (depending on what it is applied, we can infer if we actually use $F^{Obj}$ or $F_{A,B}$). 
\end{definition}

\begin{definition}[Product category] Given two categories $\category[1]$ and $\category[2]$, we can define a category $\category[1] \with \category[2]$ where: \begin{itemize}
\item $Obj(\category[1] \with \category[2]) :=\{(A_1, A_2) \ | \ A_1 \in Obj(\category[1]) \text{ and  } B_1 \in Obj(\category[2]) \}$ 
\item $(\category[1] \with \category[2])((A_1, A_2), (B_1, B_2)) := \{(f, g) \ | \ f \in \category[1](A_1, B_1) \text{ and } g \in \category[2](A_2, B_2)\}$
\item The identity is defined as $id_{(A_1, A_2)} := (id^{\category[1]}_{A_1}, id^{\category[2]}_{A_2})$ and the composition is defined as $(g_1, g_2) \circ (f_1, f_2) := (g_1 \circ_{\category[1]} f_1, g_2 \circ_{\category[2]} f_2)$. We can check that this is indeed a category.
\end{itemize}
\end{definition}


\begin{definition}[Endofunctor, bifunctor] We call an \emph{endofunctor} any functor of shape $F : \category \arrow \category$. We call a \emph{bifunctor} any functor of shape $F : \category[1] \with \category[2] \arrow \category$. 
\end{definition}

A very powerful and intuitive way of showing equalities between morphisms in categories is to use what are called \emph{commutative diagrams}. For example, the equality $h = g \circ h$ can be represented as the diagram.
\begin{center}
\begin{tikzcd}
X \arrow[r, "f"] \arrow[rd, "h"'] & Y \arrow[d, "g"] \\
                                  & Z               
\end{tikzcd}
\end{center}
More formally, a diagram is an oriented graph where vertices are indexed by objects and every edges from $X$ to $Y$ are indexed by a morphism of $\category(X,Y)$. A diagram states that for any objects $X,Y$ and any path from $X$ to $Y$, the morphisms obtained by the successive compositions are equal.

Diagrams can be combined, doing what is called a \emph{diagram chase}. The diagram below is an example of a diagram chase.
\begin{center}
\begin{tikzcd}
X_1 \arrow[r, "f_1"] \arrow[d, "h_X"'] & Y_1 \arrow[d, "h_Y"] \arrow[r, "g_1"] & Z_1 \arrow[d, "h_Z"] \\
X_2 \arrow[r, "f_2"']                  & Y_2 \arrow[r, "g_2"']                 & Z_2                 
\end{tikzcd}
\end{center}
This chase states that $g_2 \circ f_2 \circ  h_X = h_Z \circ g_1 \circ f_1$ using the fact that $f_2 \circ h_X = h_Y \circ f_1$ and $g_2 \circ h_Y = h_2 \circ g_1$. 

\begin{definition}[Natural transformation] A natural transformation $\alpha$ between two functors $F, G : \category[1] \arrow \category[2]$, written $\alpha : F \naturalTrans G$, is a family of morphisms $(\alpha_X)_{X \in Obj(\category[1])}$ where $\alpha_X \in \category[2](F X, G X)$ such that for any objects $X,Y$ of $\category[1]$ and $f \in \category[1](X, Y)$, the following diagram commutes.
\begin{center}
\begin{tikzcd}
FX \arrow[r, "F(f)"] \arrow[d, "\alpha_X"'] & FY \arrow[d, "\alpha_Y"] \\
GX \arrow[r, "G(f)"']                       & GY                      
\end{tikzcd}
\end{center}
When the morphisms are all isomorphisms, $\alpha$ is called a natural isomorphism.
\end{definition}

\end{document}